\documentclass[preprint]{elsarticle}

\usepackage{amsmath,amsfonts,subfigure,amssymb,color}
\usepackage{pta}

\journal{Theoretical Computer Science}
\allowdisplaybreaks

\begin{document}

\begin{frontmatter}

\title{Expected Reachability-Time Games}

\author[add1]{Vojt\v{e}ch Forejt} 
\author[add1]{Marta Kwiatkowska} 
\author[add3]{Gethin Norman} 
\author[add4]{Ashutosh Trivedi} 

\address[add1]{Department of Computer Science, University of Oxford, UK.}
\address[add3]{School of Computing Science, University of Glasgow, UK.} 
\address[add4]{Department of Computer Science and Engineering, IIT Bombay, India.} 

\begin{abstract}
Probabilistic timed automata are a suitable formalism to model systems with real-time, nondeterministic and probabilistic behaviour. We study two-player zero-sum games on such automata where the objective of the game is specified as the expected time to reach a target. The two players---called player Min and player Max---compete by proposing timed moves simultaneously and the move with a shorter delay is performed. The first player attempts to minimise the given objective while the second tries to maximise the objective. We observe that these games are not determined, and study decision problems related to computing the upper and lower values, showing that the problems are decidable and lie in the complexity class NEXPTIME $\cap$ co-NEXPTIME.
\end{abstract}


\end{frontmatter}

\section{Introduction}

\noindent
Two-player zero-sum games on finite automata, as a mechanism for supervisory
controller synthesis of discrete event systems, were introduced by Ramadge and
Wonham~\cite{RW89}. 
In this setting the two players---called $\mMin$ and $\mMax$---represent the
\emph{controller} and the \emph{environment}, and controller synthesis
corresponds to finding a winning (or optimal) strategy of the controller for
some given performance objective. 
Timed automata~\cite{AD94} extend finite automata by providing a
mechanism to model real-time behaviour,
while priced timed automata are timed automata with (time-dependent) prices attached to the locations of the automata.
If the game structure or objectives are dependent on time or price, e.g.\ when the objective corresponds
to completing a given set of tasks within some deadline or within some cost, then games on timed
automata are a well-established approach for controller synthesis, see e.g.\
\cite{AM99,ABM04,BCFL04,BC+07,CJLRR09}.

In this paper we extend the above approach to a setting that is quantitative
in terms of both timed and probabilistic behaviour.  
Probabilistic behaviour is important in modelling, e.g., faulty or unreliable
components, the random coin flips of distributed communication and security
protocols, and performance characteristics.  
We consider an extension of probabilistic time automata (PTA)
\cite{KNSS02,Jen96,Bea03}, a model for real-time systems exhibiting
nondeterministic and probabilistic behaviour.

In our model, called probabilistic timed game arena (PTGA), a token is placed on a configuration of a PTA 
and a play of the game corresponds to both players proposing a timed move of the
PTA, i.e.\ a time delay and action under their control (we assume each action of
the PTA is under the control of exactly one of the players).  
Once the players have made their choices, the timed move with the shorter
delay is performed and the token is moved according to the probabilistic transition
function of the PTA.    
Intuitively, players $\mMin$ and $\mMax$ represent two different forms of non-determinism,
called \emph{angelic} and \emph{demonic}.
To prevent the introduction of a third form, we assume the move of $\mMax$ (the
environment) is taken if the delays are equal. 
The converse can be used without changing the presented results.

Players $\mMin$ and $\mMax$ choose their moves in order to minimise and maximise,
respectively, the objective function.
The \emph{upper value} of a game is the
minimum expected time that $\mMin$ can ensure, while
the \emph{lower value} of a game is the maximum expected value
that $\mMax$ can ensure.  A game is \emph{determined} if the lower and
upper values are 
equal, and in this case the \emph{optimal value} of the game exists and equals
the upper and lower values.

The objectives frequently studied include reachability, which
asks for certain locations to be eventually reached, safety, which asks for a given target set
to be avoided, or more complex properties, expressed using a formula of a linear temporal logic.
The objective function is then an indicator function saying whether the property is satisfied on
a play, and the expected value then corresponds to the \emph{probability} of the property being
true.
In our paper we are interested in a more complex setting and study \emph{reachability-time} time objectives, which
express the {\em expected time} to reach a given target set. These objectives have many practical applications, e.g., in job-shop scheduling,
where machines can be faulty or have variable execution time, and both routing
and task graph scheduling problems. 
For real-life examples relevant to our setting, see
e.g.\ \cite{AH+09,CJLRR09}. The reachability-time objectives are a special case of {\em weight} or
{\em price} objectives in which different numbers are assigned to locations,
and the value of the objective function
depends on the respective numbers and the time spent in the locations; in our setting, the numbers
are fixed to be $1$ and the objective function simply sums the times spent in for each location.
Computing properties related to price functions often leads to undecidability, even in non-probabilistic setting~\cite{BBR05,BBM06}. Studying simpler properties is thus
motivated by the desire to obtain decidable properties while still being able to study sufficiently
complex class of properties.

\subsection{Contribution}
\noindent
We demonstrate the decidability of the problem of whether the upper (lower, or the optimal
 when it exists) value of a game with reachability-time objectives is at most a given bound. Our proofs
immediately yield a NEXPTIME $\cap$ co-NEXPTIME complexity bound. To our best knowledge,
this is the first decidability result for stochastic games on timed automata in which the
objective concerns a random variable that takes non-binary values.

Our approach is based on extending the boundary region graph construction for timed
automata~\cite{JT07} to PTGAs and demonstrating that the reachability-time problem can be
reduced to the same problem on the boundary region graph. 
In particular, our proof aims to show that the limit of the step-bounded
value functions in the timed automata and boundary
region graph also coincide.

Generic results exist that allow one to prove that step-bounded values converge
to the step-unbounded value, but to the best of our knowledge none are readily
applicable in our setting where the state space is uncountable and little is known
a priori about the value functions. For example, Banach fixpoint theorem requires the value iteration function (that takes a $n$-step value function and returns
the $n+1$-step value function) to be a contraction on an underlying metric space,
and it appears difficult to devise the metric space so that the contraction
property is easily obtained.
Another possible proof direction is Kleene fixpoint theorem, which requires
Scott-continuity on the value functions, which again is a property that
is difficult to establish in our setting.
We are able to partly rely on the Knaster-Tarski fixpoint theorem which characterises the set
of fixpoints, but it is not strong enough to prove the convergence itself, for reasons
similar to the ones above.
Several other theorems such as Brouwer fixpoint theorem or Kakutani fixpoint theorem
are generally not suitable for proving properties that we require in
turn-based stochastic games. 

Hence, to prove that the limit of the step-bounded value functions
is the desired value, we need to take a tailor-made approach.
We first inductively show that, when the number
of steps is bounded, then the value functions in timed automata and boundary
region graph coincide and are non-expansive within a region. Here we make use of
\emph{quasi-simple functions} which generalise simple functions, previously
used by Asarin and Maler in the study of games over non-probabilistic timed automata \cite{AM99}.
Then, using the non-expansiveness property, we show that the limit of the step-bounded
value functions in the timed automata and boundary
region graph also coincide. In this part we use Knaster-Tarski fixpoint theorem.

The definition of quasi-simple functions is a central component of our proof, as it is
strong enough to enable us to utilise an approach used in proofs of fixpoint theorems, but on the other hand
general enough to capture the values of reachability-time objectives. We believe that
it can serve as a step from simple functions towards functions describing even more complex
but still decidable objectives.

\subsection{Related Work}
\noindent
Hoffman and Wong-Toi~\cite{HW92} were the first to define and solve the optimal
controller synthesis problem for timed automata. 
For a detailed introduction to the topic of qualitative games on timed
automata, see e.g.\ \cite{AMP95}.  
Asarin and Maler~\cite{AM99} initiated the study of quantitative games on timed
automata by providing a symbolic algorithm to solve reachability-time objectives. 
The works of \cite{BHPR07} and \cite{JT07} show that the decision problem for such games over timed automata with at least two clocks is EXPTIME-complete.  
The tool UPPAAL Tiga~\cite{BC+07} is capable of solving reachability and safety
objectives for games on timed automata.
Jurdzi{\'n}ski and Trivedi~\cite{JT08b} show the EXPTIME-completeness for
average-time games on automata with two or more clocks.

A natural extension of games with reachability-time objectives are games on priced timed automata
where the objective concerns the cumulated price of reaching a target.
Both \cite{ABM04} and \cite{BCFL04} present semi-algorithms
for computing the value of such games for linear prices.  
In~\cite{BBR05} the problem of checking the existence of
optimal strategies is shown to be undecidable, with \cite{BBM06} showing
undecidability holds even for three clocks and stopwatch prices.   

As for two-player quantitative games on PTAs,    
for a significantly different model of stochastic timed games, deciding whether
a target is reachable within a given bound is undecidable \cite{BF09}.
In~\cite{DBLP:conf/concur/BrazdilKKKR10},
continuous-time games are verified against time-automata objectives, giving rise to systems
whose semantics is related to the ones of \cite{BF09}.
The work of~\cite{timed-robust} studies probability of satisfying B\"uchi objectives in
a timed game where perturbations of probabilities can take place, and~\cite{DBLP:conf/fsttcs/BrazdilHKKR12}
studies games on interactive Markov chains which are modelled as a game extension of timed automata.

Regarding one-player games on PTAs, in \cite{BCJ09} the problem of deciding whether a target can be reached
within a given price and probability bound is shown to be undecidable for
priced PTAs with three clocks and stopwatch prices.  
The work of \cite{JKNT09} shows that the problem becomes decidable when the price functions
are of a restricted form. In \cite{JKN15}, simple functions are extended to devise a symbolic algorithm
for computing minimum expected time to reach a target in one-player games on PTAs; the extension
differs fundamentally from our quasi-simple functions.
We also mention the approaches for analysing unpriced probabilistic
timed automata against temporal logic specifications 
based on the region graph \cite{KNSS02,Jen96} and either forwards
\cite{KNSS02} or backwards \cite{KNSW07} reachability. 
The complexity of performing such verification is studied in \cite{LS07} for almost-sure reachability, and in \cite{JLS08}
for PCTL properties and a restricted number of clocks. Finally,~\cite{DBLP:conf/qest/BrazdilKKNR15} deals with a model similar to PTAs in which
time evolves continuously and controllable ``fixed delay'' events are introduced.

\vskip8pt 
\noindent
A preliminary version of the work was published in
conference proceedings \cite{FKNT10a}.
The result presented in \cite{FKNT10a} required
an assumption on the structure of the PTAs that enforced a terminal state to be reached almost surely
under any pair of strategies. In this paper we \emph{lift this restriction}
and consider arbitrary PTAs. Further, the proofs in \cite{FKNT10a} contain
a significant flaw which required major changes to be made to the proof, also for the restricted case.
Thus, although the high-level idea behind the proof (bounding the difference of values
for two configurations whose clock values are close to each other) stays the same, the
actual steps of the proof changed significantly.
Note that, although \cite{FKNT10a} also introduces quasi-simple functions, the definition
used here is different (and not equivalent). Most notably, our proofs here use a much more
``constructive'' approach when defining value functions.

\subsection{Outline}

\noindent
The structure of the paper is the following. 
In Section~\ref{sec:erg} we introduce Stochastic Games Arenas, which serve as
semantics for games on PTAs. 
Games on PTAs are then introduced in Section~\ref{dec:ertg}, and Section~\ref{subsec:boundary-region-graph}
defines boundary region abstraction, which plays a fundamental role in our proofs.
Section~\ref{sec:reach} provides the proofs for the main result. 

\section{Stochastic Game Arena}\label{sec:erg}

\noindent
We now introduce a general notion of stochastic game arenas (SGAs), which will later serve as semantics for the model we study.
The reader may notice that our definition of a stochastic game arena
differs from the standard concurrent stochastic game arena~\cite{FV97, dAM04}. 
However, as we shall demonstrate later, it
captures precisely the semantics of probabilistic timed game arena.
In addition, presenting the basic concepts relating to values in the general
setting of SGAs allows us to use these concepts in the context of both probabilistic timed game arenas and their abstractions.

\subsection{Stochastic Game Arena: Syntax and Semantics}

\noindent
We write $\Nat$ for the set of non-negative integers, $\Rat$ for the rational numbers,
$\Rplus$ for the non-negative reals, and $\Rplus^\infty$ for the reals with the maximum element $\infty$.
A function $f : (\Rplus^\infty)^n {\to} \Rplus^\infty$ is \emph{non-expansive}
if for any $\mathbf{x},\mathbf{y} \in (\Rplus^\infty)^n$ we have $\norm{f(\mathbf{x}) {-} f(\mathbf{y})} \le \norm{\mathbf{x} {-} \mathbf{y}}$ where
$\norm{\cdot}$ is the max norm, i.e. $|(x_1,\ldots,x_n)| = \max_{1\le i \le n} | x_i |$.
A \emph{discrete probability distribution}, or just distribution, over a (possibly uncountable) set $Q$ is 
a function $\dis : Q {\to} [0, 1]$ such that 
$\sum_{q \in Q} \dis(q) = 1$ and $\supp(\dis) \rmdef \set{q \in Q \mid  \dis(q)
  {>} 0}$ is at most countable. 
Let $\DIST(Q)$ denote the set of all discrete distributions over $Q$.
We say a distribution ${\dis \in \DIST(Q)}$ is a \emph{point distribution}
if $\dis(q) {=} 1$ for some $q \in Q$.
Given a set $Q$ and two functions $f,f': Q {\to} \Rplus^\infty$, we state $f {\le} f'$ when $f(q) {\le} f'(q)$ for all $q\in Q$.
A function $f: Q {\to} \Rplus$ is a {\em convex combination} of functions
$f_1,\ldots, f_n : Q {\to} \Rplus$ if there are non-negative coefficients $p_1 ,\ldots, p_n$
such that $\sum_{i=1}^n p_i = 1$ and we have $f(q) = \sum_{i=1}^n p_i {\cdot} f_i(q)$ for all $q\in Q$.
\begin{definition}[Stochastic Game Arena (SGA)]\label{sga-def}
  A stochastic game arena is given by a tuple  $\gam{=}(S, A_\mMIN, A_\mMAX,
  p_\mMIN, p_\mMAX, \winner, \tau_\mMIN, \tau_\mMAX)$ where: 
  \begin{itemize}
  \item
    $S$ is a possibly uncountable set of \emph{states};
  \item 
    $A_\mMIN$ and $A_\mMAX$ are possibly uncountable sets of \emph{actions controlled} by players $\mMin$ and $\mMax$ respectively, 
    and $\bot$ is a
    distinguished action such that $A_\mMIN\cap A_\mMAX = \{\bot\}$;
  \item
    $p_\mMIN : (S {\times} A_\mMIN) {\pto} \DIST(S)$ and $p_\mMAX : (S {\times}
    A_\mMAX) {\pto} \DIST(S)$ are the probabilistic transition (partial) functions
    for players $\mMin$ and $\mMax$ respectively, such that $p_\mMIN(s, \bot)$ and 
    $p_\mMAX(s, \bot)$ are undefined for all $s \in S$, and for any $s\in S$ 
    either there exists $\alpha\in A_\mMIN$ such that $p_\mMIN(s,\alpha)$ is defined 
    or there exists $\alpha\in A_\mMAX$ such that
    $p_\mMAX(s,\alpha)$ is defined;
  \item $\winner : (A_\mMIN {\times} A_\mMAX) {\to} (A_\mMIN \cup A_\mMAX)$ is a function
   specifying which of the actions chosen by the players takes place, requiring that for any $(\alpha,\beta) \in A_\mMIN {\times} A_\mMAX$ we have $\winner(\alpha,\beta) \in \{ \alpha,\beta\}$, and moreover $\winner(\alpha,\beta)$
   is never equal to $\bot$ unless $\alpha{=}\beta{=}\bot$;
  \item 
    $\tau_\mMIN : (S {\times} A_\mMIN) {\pto} \Rplus$ and $\tau_\mMAX : (S {\times}
    A_\mMAX) {\pto} \Rplus$ are the time delay (partial) functions for players $\mMin$ and $\mMax$ respectively, specifying the \emph{delay} associated with
    performing an action in a state.
\end{itemize}
\end{definition}
Note that SGAs introduced above are more general than classical stochastic games, in particular SGAs
contain information about the time delays of actions.

We say that an SGA is \emph{finite} if $S$, $A_\mMIN$ and $A_\mMAX$ are finite. 
For any state $s \in S$, we let $A_\mMIN(s)$ denote the set of actions available to
player $\mMin$ in $s$, i.e.,  the actions $\alpha \in A_\mMIN$ for which $p_\mMIN(s,\alpha)$
is defined, letting $A_\mMIN(s){=}\{\bot\}$ if no such action exists.
Similarly, $A_\mMAX(s)$ denotes the actions
available to player $\mMax$ in $s$ and we let $A(s) {=} A_\mMIN(s) {\times} A_\mMAX(s)$. From the conditions required of the probabilistic transition functions of the players, we have $(\bot,\bot) \not\in A(s)$ for all $s \in S$. 

A game on SGA $\gam$ starts with a token in an {\em initial state} $s\in S$ and players $\mMin$
and $\mMax$ construct an infinite play by repeatedly choosing enabled actions, and then
moving the token to a successor state determined by the probabilistic
transition function of the player proposing the action that is favoured by the $\winner$ function. Formally, we introduce the following auxiliary definition for an SGA.
\begin{definition}[Probabilistic transition function of an SGA]\label{trans-def} 
  For any stochastic game arena $\gam{=} (S, A_\mMIN, A_\mMAX,
  p_\mMIN, p_\mMAX, \tau_\mMIN, \tau_\mMAX)$ the \emph{probabilistic transition function} of $\gam$ is given by the partial function
  $p : (S {\times} A_\mMIN {\times} A_\mMAX) {\pto} \DIST(S)$ where for any  $s \in S$, $\alpha \in A_\mMIN$ and $\beta \in A_\mMAX$: 
  \begin{align*}
    p(s,\alpha,\beta) = & \; \left\{ \begin{array}{cl}
        \mbox{undefined} & \mbox{if} \; (\alpha,\beta) \not\in A(s)  \\
        p_\mMIN(s,\alpha) & \mbox{if} \;  \winner(\alpha,\beta) {=} \alpha\\
        p_\mMAX(s,\beta) & \mbox{otherwise.} \end{array} \right. 
  \end{align*}
\end{definition} 
Using this definition, if we are in state $s$ and the action pair $(\alpha, \beta) \in A(s)$ is
chosen by the players, then the probability of making a transition to $s'$ equals $p(s, \alpha, \beta)(s')$.
We similarly define the time delay function $\tau$ of the SGA $\gam$ by
  \begin{align*}
    \tau(s,\alpha,\beta) \rmdef & \; \left\{ \begin{array}{cl}
        \mbox{undefined} & \mbox{if} \; (\alpha,\beta) \not\in A(s)  \\
        \tau_\mMIN(s,\alpha) & \mbox{if} \;  \winner(\alpha,\beta) {=} \alpha\\
        \tau_\mMAX(s,\beta) & \mbox{otherwise.} \end{array} \right. 
  \end{align*}
A transition of $\gam$ is a tuple $(s, (\alpha, \beta), s')$ such that
$p(s, \alpha, \beta)(s') {>} 0$ and a play of $\gam$ is a finite or infinite sequence   
\[
\seq{s_0, (\alpha_1, \beta_1) , s_1 ,  (\alpha_2,\beta_2) , \ldots , s_{i} , (\alpha_{i+1},\beta_{i+1}) , s_{i+1} , \ldots }
\]
such that $(s_i, (\alpha_{i+1}, \beta_{i+1}), s_{i+1})$ is a transition for all
$i {\geq} 0$. The length of a play $\rho$, denoted $\len(\rho)$, is defined as the number of transitions appearing in the play.
For a finite play $\rho {=}  \seq{s_0, (\alpha_1, \beta_1), s_1, \ldots, (\alpha_k, \beta_k), s_k}$, let
$\LAST(\rho)$ denote the last state $s_k$ of the play. 
We write $\RUNS$ ($\FRUNS$) for the sets of infinite (finite) plays in $\gam$
and $\RUNS(s)$ ($\FRUNS(s)$) for the sets of infinite (finite) plays
starting from $s \in S$.   
\begin{definition}[SGA Strategy]
Let $\gam{=} (S, A_\mMIN, A_\mMAX,p_\mMIN, p_\mMAX, \tau_\mMIN, \tau_\mMAX)$ be a SGA. A \emph{strategy} of $\mMin$ is a function $\mu : \FRUNS {\to} A_{\mMIN}$ such that $\mu(\rho) \in A_\mMIN(\LAST(\rho))$ for all finite plays $\rho \in \FRUNS$. A strategy $\chi$ of $\mMax$ is defined analogously and we let $\Sigma_\mMIN$ and $\Sigma_\mMAX$ denote the sets of strategies of $\mMin$ and $\mMax$, respectively.
\end{definition}
For any finite play, a strategy of $\mMin$ ($\mMax$) returns an action
available to $\mMin$ ($\mMax$) in the last state of the play.

For a SGA $\gam$, state $s$ of $\gam$ and strategy pair $(\mu,\chi) \in  \Sigma_\mMIN {\times}
\Sigma_\mMAX$, let $\RUNS^{\mu, \chi}(s)$ ($\FRUNS^{\mu, \chi}(s)$) denote the set of infinite (finite) plays in which 
$\mMin$ and $\mMax$ play according to $\mu$ and $\chi$, respectively.  
Given a finite play $\rho \in \FRUNS^{\mu, \chi}(s)$, a basic cylinder set $\cyl(\rho)$
is the set of infinite plays in $\RUNS^{\mu, \chi}(s)$ for which $\rho$ is a prefix.
Using standard results from probability theory~\cite{NS04} we
can construct a probability space $(\RUNS^{\mu, \chi}(s), \mathcal{F}^{\mu, \chi}(s), \PROB^{\mu, \chi}_s)$
where $\mathcal{F}^{\mu, \chi}(s)$ is the smallest $\sigma$-algebra generated by the basic cylinder sets and $\PROB^{\mu, \chi}_s : \mathcal{F} {\to} [0,1]$ is the unique probability measure such that for any
finite play $\rho {=}  \seq{s_0, (\alpha_1, \beta_1), s_1, \ldots, s_{k-1} , (\alpha_k, \beta_k), s_k} \in \FRUNS^{\mu, \chi}(s)$:
\[ \begin{array}{c}
 \PROB^{\mu, \chi}_s(\mathit{Cyl}(\rho)) = \prod_{i=1}^k p(s_{i-1},\alpha_i,\beta_i)(s_i)
\end{array} \]
where $\mathit{pre}(\rho,i){=}\seq{s_0, (\alpha_1, \beta_1), s_1, \ldots, s_{i-1} , (\alpha_i, \beta_i) , s_{i}}$ for all $i {<} k$.

Given a 
\emph{real-valued random variable} $f : \RUNS {\to} \Rplus^\infty$, the expression 
$\eE^{\mu, \chi}_{s}(f)$ denotes the expected value of $f$ with respect to the
probability measure $\PROB^{\mu, \chi}_s$.      

We extend $\PROB^{\mu, \chi}_{s}(f)$ also to the cases where the game is assumed
to start from a finite play $\rho {=} \seq{s_0, (\alpha_1, \beta_1), s_1, \ldots, (\alpha_k, \beta_k), s_k}$ as opposed to a state $s$, and we let
$\PROB^{\mu, \chi}_{\rho} = \PROB^{\mu_\rho, \chi_\rho}_{\LAST(\rho)}$,
where the strategy $\mu_\rho$ is defined from $\mu$ by
\[
 \mu_\rho(\rho') =
\mu(\seq{s_0, (\alpha_1, \beta_1), s_1, \ldots, (\alpha_k, \beta_k), s_k, (\alpha'_{k+1},\beta'_{k+1}), s'_{k+1},\ldots (\alpha_\ell,\beta_\ell),s_\ell})
\]
for all $\rho'{=}\seq{s'_k,\ldots (\alpha'_k,\beta'_k), s'_{k+1},\ldots (\alpha_\ell,\beta_\ell),s_\ell}$
such that $s_k{=}s'_k$, and $\mu_\rho(\rho')$ is defined arbitrarily otherwise;
the strategy $\chi_\rho$ is defined analogously. We then also use $\eE^{\mu, \chi}_{\rho}$ defined with respect to $\PROB^{\mu, \chi}_{\rho}$.

\subsection{Reachability-time objective in Stochastic Game Arena}

\noindent 
We now define the reachability-time objective for plays of SGAs.
\begin{definition}\label{obj-def}
For an SGA $\gam$ and target set of states $F$ of $\gam$, the \emph{(finite-horizon) $n$-step reachability-time objective}
associated with an infinite play $\rho {=} \seq{s_0, (\alpha_1, \beta_1) , s_1 , \ldots}$ is
given by:
\[ \begin{array}{c}
 \REACHPRICE^n_{F}(\rho) = \sum_{i=1}^{k} \tau(s_{i-1}, \alpha_i, \beta_i)
\end{array} \]
where $k {=} \min\{i \in \Nat \mid s_i \in F\}$ if $s_j \in F$ for some $j {<} n \in \Nat$
and $k{=} n $ otherwise.
Furthermore, the \emph{(infinite-horizon) reachability-time objective} (with target set $F \subseteq S$)
associated with an infinite play $\rho$ is
given by:
\[ \begin{array}{c}
\REACHPRICE_{F}(\rho) = \lim_{n\rightarrow \infty} \REACHPRICE^n_{F}(\rho) \, .
\end{array} \]
\end{definition}
In the definition of the infinite horizon objective the limit always exists, but it can be infinite. To simplify notation, we often omit the target set $F$ when it is clear from the context.

In our games on an SGAs players $\mMin$ and $\mMax$ move a token
along the edges in order to minimise and maximise, respectively, the ($n$-step) reachability-time objective function. 
Formally, for an SGA $\gam$ and an objective
$\REACHPRICE^n$
we define lower and upper value with respect to $\REACHPRICE^n$ for $\gam$ in state
$s \in S$ by 
\begin{align*}
 \LVAL_\gam^n(s) \rmdef & \; \sup\nolimits_{\chi \in \Sigma_\mMAX} \inf\nolimits_{\mu \in \Sigma_\mMIN}
 \eE^{\mu, \chi}_s(\REACHPRICE^n)  \\ 
 \UVAL_\gam^n(s) \rmdef & \; \inf\nolimits_{\mu \in \Sigma_\mMIN} \sup\nolimits_{\chi \in \Sigma_\mMAX}
 \eE^{\mu, \chi}_s(\REACHPRICE^n)
\end{align*}
respectively. Similarly, for an objective $\REACHPRICE$ we define the lower and upper values:
\begin{align*}
 \LVAL_\gam(s) \rmdef & \; \sup\nolimits_{\chi \in \Sigma_\mMAX} \inf\nolimits_{\mu \in \Sigma_\mMIN}
 \eE^{\mu, \chi}_s(\REACHPRICE) \\ 
 \UVAL_\gam(s) \rmdef & \; \inf\nolimits_{\mu \in \Sigma_\mMIN} \sup\nolimits_{\chi \in \Sigma_\mMAX}
 \eE^{\mu, \chi}_s(\REACHPRICE).
\end{align*}
In the cases when the lower and upper values coincide, we denote this value simply as $\VAL_\gam^n(s)$ or $\VAL_\gam(s)$
and say that the corresponding game is \emph{determined}. We omit $\gam$ if it is clear from the context, e.g. we write simply $\UVAL$ instead of $\UVAL_\gam$.

For $\mu \in \Sigma_\mMIN$, $\chi \in \Sigma_\mMAX$ and $s \in S$,  let
\[
\VAL_\gam(s, \mu) \rmdef \sup\nolimits_{\chi' \in
  \Sigma_\mMAX} \eE^{\mu,\chi'}_s(\REACHPRICE) 
  \; \; \mbox{and} \;\;
\VAL_\gam(s, \chi) \rmdef \inf\nolimits_{\mu' \in  
\Sigma_\mMIN} \eE^{\mu',\chi}_s(\REACHPRICE)  \, .
  \]
We say $\mu$ is \emph{optimal} (or \emph{$\varepsilon$-optimal}), if
$\VAL_\gam(s, \mu)  {=}  \UVAL_\gam(s)$  
(or $\VAL_\gam(s, \mu)  \leq \UVAL_\gam(s) {-} \varepsilon$)
for all $s \in S$.
Furthermore, $\chi$ is \emph{optimal} (or \emph{$\varepsilon$-optimal}), if
$\VAL_\gam(s, \chi)  {=}  \UVAL_\gam(s)$  
(or $\VAL_\gam(s, \chi)  \geq \LVAL_\gam(s) {-} \varepsilon$)
for all $s \in S$.    
If $\gam$ is determined, then each player has an $\varepsilon$-optimal strategy
for all $\varepsilon{>}0$.

Since we will consider two-player games on SGAs that are not
determined, we are interested in the following problem with respect
to the upper value of a game. 
\begin{definition}
  Given an SGA $\gam$, initial state $s \in S$, reachability-time objective
  and value $B \in \Rat$, the corresponding game {\em reachability-time problem} is to decide
  whether $\UVAL(s) \leq B$. 
\end{definition} 
All results presented in the paper are still valid when replacing the upper value with the lower value. The following is a well-known result.
\begin{theorem}[\cite{Con92,ZP96}]
  \label{theorem:finite-games}
  The reachability-time problem for infinite-horizon objectives over finite SGAs is in NP $\cap$ co-NP.
\end{theorem}
Efficient algorithms exist to solve the problem over finite SGAs, e.g. using value iteration~\cite{CH08,Con93}.

\subsection{Optimality Equations for SGAs}\label{reach-opt-sect}
\noindent
We now introduce optimality equations for reachability objectives over SGAs. For the remainder of this section we fix an SGA $\gam{=}(S, A_\mMIN, A_\mMAX, p_\mMIN, p_\mMAX, \tau_\mMIN, \tau_\mMAX)$
and a target set $F \subseteq S$.
\begin{definition}\label{bell-def}
The Bellman-style equations for $n$-step reachability time objective are given as follows: $\UVAL^n(s){=}0$ whenever $n{=}0$ or $s \in F$, and
for $n{\geq}0$ and $s\not\in F$:
\[ \begin{array}{c}
 \UVAL^{n+1}(s) =
    \inf\limits_{\alpha \in A_\mMIN(s)} \sup\limits_{\beta\in A_\mMAX(s)} \left\{  
    \tau(s,\alpha,\beta) + \sum_{s' \in S} p(s, \alpha, \beta)(s')  {\cdot}
    \UVAL^{n}(s')\right\}
\end{array} \]
\end{definition}
The correctness of these equations can be easily obtained from the fact that for any $n{\geq}0$, $s\not\in F$, path $\rho$ with $\LAST(\rho) {=} s$ and strategies $\mu$ and $\chi$, where $\alpha{=}\winner(\mu(\rho),\chi(\rho))$, by definition of $\eE^{\mu,\chi}_\rho$ we have:
\begin{align*}
\lefteqn{\eE^{\mu,\chi}_\rho (\REACHPRICE_{F}^{n+1})
   = \eE^{\mu_\rho,\chi_\rho}_{s} (\REACHPRICE_{F}^{n+1})} \\
   =&\; \int_{\rho'\in\RUNS^{\mu_\rho, \chi_\rho}(s)} \REACHPRICE_{F}^{n+1}(\rho')\, d\PROB^{\mu_\rho, \chi_\rho}_{s}
    \tag{by definition of expectation}\\
   =&\;  \sum_{s'\in S} \Bigg( \int_{\bar\rho{\in} \RUNS^{\mu_{\rho \alpha \ssp{s}}, \chi_{\rho \alpha \ssp{s}}}(s')} p(s, \mu_\rho(s),\chi_\rho(s))(s')\\
   &\qquad \cdot  \Big(\tau(s,\mu_\rho(s),\chi_\rho(s)) + \REACHPRICE_{F}^{n}(\bar\rho)\Big)\,d\PROB^{\mu_{\rho \alpha \ssp{s}}, \chi_{\rho \alpha \ssp{s}}}_{s'}(\bar\rho) \Bigg)
    \tag{by definition of $\RUNS^{\mu_\rho, \chi_\rho}(s)$, $\PROB^{\mu_\rho, \chi_\rho}_{s}$ and $\REACHPRICE_{F}^{n+1}$}\\
   =&\;  \tau(s,\mu_\rho(s),\chi_\rho(s)) \\
   &\;  + \sum_{s'\in S}\!p(s, \mu_\rho(s),\chi_\rho(s))(s') \cdot \left( \int_{\bar\rho{\in}\RUNS^{\mu_{\rho \alpha \ssp{s}}, \chi_{\rho \alpha \ssp{s}}}(s')} \REACHPRICE_{F}^{n}(\bar\rho)\,d \PROB^{\mu_{\rho \alpha \ssp{s}}, \chi_{\rho \alpha \ssp{s}}}_{s'}(\bar\rho) \right)
    \tag{rearranging}\\
   =&\;  \tau(s,\mu(\rho),\chi(\rho)) + \sum_{s'\in S} p(s, \mu(\rho),\chi(\rho))(s') \cdot \eE^{\mu,\chi}_{\rho\, \alpha\, s'} (\REACHPRICE_{F}^{n})
    \tag{by properties of $\mu_\rho$, $\chi_\rho$ and definition of expectation}
\end{align*}
Let us now turn to the equations for infinite-horizon objectives.
\begin{definition}\label{opt-def}
A function $P : S {\to} \Rplus^\infty$ is a
solution of the optimality equations $\UOpt_\gam$, written $P \models
\UOpt_\gam$, if for any $s \in S$: 
\[
P(s) {=} \left\{
  \begin{array}{cl}
    \!\!\!\! 0 & \text{if $s \in F$}\\
    \!\!\!\!\inf\limits_{\alpha \in A_\mMIN(s)}  
    \sup\limits_{\beta\in A_\mMAX(s)} \left\{  
    \tau(s,\alpha,\beta) + \sum\limits_{s' \in S} p(s, \alpha, \beta)(s')  {\cdot}
    P(s')\right\} \!\!  & \text{if $s \not\in F$}
  \end{array} \right.
\]
and is a solution of the optimality equations $\LOpt_\gam$, written 
$P \models \LOpt_\gam$, if for any $s \in S$:  
\[
P(s) {=} \left\{
  \begin{array}{cl}
    \!\!\!\!0 & \text{if $s \in F$}\\
    \!\!\!\!\sup\limits_{\beta\in A_\mMAX(s)}  
    \inf\limits_{\alpha \in A_\mMIN(s)} \left\{ 
    \tau(s,\alpha,\beta) + \sum\limits_{s' \in S} p(s,\alpha, \beta)(s')  {\cdot}
    P(s')\right\}  & \text{if $s \not\in F$.}
  \end{array} \right.
\]
\end{definition}
To simplify the presentation, from now we will only concentrate on upper value $\UVAL$. Analogous results for the lower value follow in a straightforward manner.

Our aim is to utilise the optimality equations for $\UOpt_\gam$ and prove that $\UVAL$ and $\lim_{n\rightarrow \infty}\UVAL^n$
are equal, as an initial step towards computing or approximating $\UVAL$. Although
this equivalence can seem obvious, it is not at all trivial and, due to the uncountable nature of SGAs, it is not
possible to use results such as Kleene fixpoint theorem out of the box. In fact, in this paper we will only prove the equivalence
for a special case of SGAs (sufficient for our purpose). Nevertheless, the following two lemmas can be established for
SGAs in general.
\begin{lemma}\label{lemma:value-under-fixpoint}
  For any solution $V \models  \UOpt_\gam$ we have $\UVAL \le V$.
\end{lemma}
\begin{proof}
Consider any $\varepsilon{>}0$ and let $\mu$ be a strategy for player $\mMIN$ that, for any finite play $\rho$, selects an $\varepsilon {\cdot} 2^{-(\len(\rho)+1)}$ optimal action. For an {\em initial state} $s\in S$ and a finite play $\rho$ such that $\LAST(\rho){=}s$, it follows that: 
\begin{equation}\label{epsilon-eqn}
V(s) 
  + \varepsilon {\cdot} 2^{-(\len(\rho)+1)} \geq    \sup_{\beta\in A_\mMAX(s)} \left\{ \tau(s,\mu(\rho),\beta) + \mbox{$\sum\limits_{s'\in S}$} p(\mu(\rho),\beta)(s') {\cdot} V(s') \right\} \, .
\end{equation}
We will now show that for any path $\rho$, counter-strategy $\chi$ for $\mMAX$ and $n \in \Nat$ we have:
\begin{align}
  \eE^{\mu,\chi}_\rho (\REACHPRICE_{F}^n) & \le \; V(\LAST(\rho)) + \mbox{$\sum\nolimits_{m =\len(\rho)+1}^{n+\len(\rho)}$} \varepsilon {\cdot} 2^{-m} \label{opt-eqn} \, .
\end{align}
We prove (\ref{opt-eqn}) by induction on $n \in \Nat$. The case for $n{=}0$ follows from Definition~\ref{obj-def} and Definition~\ref{opt-def}.
  
Now suppose (\ref{opt-eqn}) holds for some $n \in \Nat$.
Consider any finite path $\rho$ where $\LAST(\rho)=s$ and counter-strategy $\chi$ for $\mMAX$. Now, if $s \in F$, then by Definition~\ref{obj-def} we have:
\[ \begin{array}{c}
\eE^{\mu,\chi}_\rho (\REACHPRICE_{F}^{n+1}) = 1 =  V(\LAST(\rho)) \leq V(\LAST(\rho)) + \sum_{m=\len(\rho)+1}^{(n+1)+\len(\rho)} \varepsilon {\cdot} 2^{-m} \, .
\end{array} \]
On the other hand, if $s \not\in F$ and letting $a{=}\winner(\mu(\rho),\chi(\rho))$, then by Definition~\ref{obj-def} and Definition~\ref{opt-def}:
  \begin{align*}
  \lefteqn{ \eE^{\mu,\chi}_\rho  (\REACHPRICE_{F}^{n+1}) = \tau(s,\mu(\rho),\chi(\rho)) + \mbox{$\sum\limits_{s'\in S}$} p(s, \mu(\rho),\chi(\rho))(s') {\cdot} \eE^{\mu,\chi}_{\rho a s'} (\REACHPRICE_{F}^{n})} \\
    &\le \; \tau(s,\mu(\rho),\chi(\rho)) + \mbox{$\sum\limits_{s'\in S}$} p(s, \mu(\rho),\chi(\rho))(s') {\cdot}\left( V(s') + \mbox{$\sum\limits_{m=\len(\rho a s')+1}^{n+\len(\rho a s')}$} \varepsilon{\cdot}2^{-m}\right)  \tag{by induction} \\
    &\le \; \tau(s,\mu(\rho),\chi(\rho)) + \mbox{$\sum\limits_{s'\in S}$} p(s, \mu(\rho),\chi(\rho))(s') {\cdot}\left( V(s') + \mbox{$\sum\limits_{m=\len(\rho)+2}^{n+\len(\rho)+1}$} \varepsilon{\cdot}2^{-m}\right)  \tag{by definition of $\len(\cdot)$} \\
    &= \; \left(\tau(s,\mu(\rho),\chi(\rho)) + \mbox{$\sum\limits_{s'\in S}$} p(s, \mu(\rho),\chi(\rho))(s') {\cdot} V(s')\right)+ \mbox{$\sum\limits_{m=\len(\rho)+2}^{(n+1)+\len(\rho)}$} \varepsilon{\cdot}2^{-m}  \tag{rearranging} \\
    &\le \; V(s) + \varepsilon{\cdot}2^{-(\len(\rho)+1)} + \mbox{$\sum\limits_{m=\len(\rho)+2}^{(n+1)+\len(\rho)}$}\varepsilon{\cdot}2^{-m} \tag{by (\ref{epsilon-eqn})} \\
    &= \; V(\LAST(\rho)) + \mbox{$\sum\limits_{m=\len(\rho)+1}^{(n+1)+\len(\rho)}$} \varepsilon{\cdot}2^{-m} \tag{rearranging.}
  \end{align*}
Since these are all the cases to consider, it follows that (\ref{opt-eqn}) holds by induction on $n$.

Letting $\rho=s$ and taking the limit of $n$ in (\ref{opt-eqn}), we have $\eE^{\mu,\chi}_s (\REACHPRICE_{F}) \le V(s) + \varepsilon$ and, since $\varepsilon$ and $\chi$ were arbitrary, it follows that $\UVAL(s) \le V(s)$ as required. \qed
\end{proof}
\begin{lemma}\label{lemma:ih-at-least-fh}
 $\UVAL \ge \lim_{n\rightarrow\infty} \UVAL^n$.
\end{lemma}
\begin{proof}
The proof follows straightforwardly from the fact that for any $n \in \Nat$ and finite play $\rho$ we have that $\REACHPRICE_{F}(\rho) \ge \REACHPRICE^n_{F}(\rho)$. \qed
\end{proof}
\section{Probabilistic Timed Game Arenas}\label{dec:ertg}

\noindent
In this section we introduce Probabilistic Timed Game Arenas (PTGAs) which extend classical timed
automata~\cite{AD94} with discrete distributions and a partition of the actions
between two players $\mMin$ and $\mMax$. 
However, before we present syntax and semantics of PTGAs, we need to introduce
the concept of clock variables and related notions. 
\subsection{Clocks, Constraints, Regions, and Zones}
\paragraph*{Clocks}
Let $\clocks$ be a finite set of \emph{clocks}.
A \emph{clock valuation} on $\clocks$ is a function
$\nu : \clocks {\to} \Rplus$ and we write $V(\clocks)$ (or just $V$ when $\clocks$ is clear from the context) for the set of clock
valuations.
Abusing notation, we also treat a valuation $\nu$ as a point in
$(\Rplus)^{|\clocks|}$. Let $\mathbf{0}$ denote the clock valuation that assigns 0 to all clocks.
If~$\nu \in V$ and $t \in \Rplus$ then we write $\nu {+} t$ for the
clock valuation defined by $(\nu {+} t)(c) \rmdef \nu(c) {+} t$ for all 
$c \in \clocks$.
For $C \subseteq \clocks$, we write $\nu_C$ for the valuation
where $\nu_C(c)$ equals $0$ if 
$c \in C$ and $\nu(c)$ otherwise.
For $X \subseteq V(\clocks)$, we write $\CLOS{X}$ for the smallest closed set
in~$V$ containing $X$.  
Although clocks are usually allowed to take arbitrary non-negative values,
for notational convenience we assume that there is an upper bound $K \in \Nat$ such that
for every clock $c \in \clocks$ we have that $\nu(c) \leq K$.

\paragraph{Clock constraints}
A \emph{clock constraint} over $\clocks$ with upper bound $K \in \Nat$ is a
conjunction of \emph{simple constraints} of the form $c \bowtie i$ or $c {-} c'
\bowtie i$, where  $c, c' \in \clocks$, $i \in \Nat$, $i {\leq} K$, and 
${\bowtie} \in \{{<}, {>}, {=}, {\leq} , {\geq} \}$. 
For $\nu \in V(\clocks)$ and $K \in \Nat$, let $\CC(\nu, K)$ be the set of clock constraints with
upper bound $K$ which hold in~$\nu$, i.e.\ those constraints that resolve to
$\mathtt{true}$ after substituting each occurrence of a clock $x$ with $\nu(x)$.

\paragraph*{Clock regions}
Every clock region is an equivalence class of the
indistinguishability-by-clock-constraints relation, and vice versa.   
For a given set of clocks $\clocks$ and upper bound $K \in \Nat$ on
clock constraints, a \emph{clock
  region} is a maximal set $\region {\subseteq} V(\clocks)$ such that 
$\CC(\nu, K) {=} \CC(\nu', K)$ for all $\nu, \nu' \in \region$. 
For the set of clocks $\clocks$ and upper bound $K$ we write
$\Rr(\clocks, K)$ for the corresponding finite set of clock regions. 
We write $[\nu]$ for the clock region of $\nu$. 
If $\region {=} [\nu]$, write $\region_C$ for $[\nu_C]$; this
definition is well-defined, since for any clock
valuations $\nu$ and $\nu'$ if $[\nu]{=}[\nu']$ then $[\nu_C]{=}[\nu'_C]$.

\paragraph*{Clock zones}
A \emph{clock zone} is a convex set of clock valuations, which is a union of a
set of clock regions.  
We write $\zones(\clocks,K)$ for the set of clock zones over the set of clocks
$\clocks$ and upper bound $K$. 
Observe that a set of clock valuations is a clock zone if and only if it is
definable by a clock constraint. Although more than one clock constraint can represent the same zone, for any clock zone $\zeta$, there exists an $O(|\clocks|^3)$ algorithm to compute the (unique) canonical clock constraint of $\zeta$ \cite{Dil89}. We therefore interchange the semantic and syntactic interpretation of clock zones.
\\ \\ \noindent
When the set of clocks and upper bound is clear from the context we write  $\Rr$ and $\zones$ for the set of regions and zones respectively.

\subsection{Probabilistic Timed Game Arena: Syntax}

\noindent
For the remainder of the paper we fix a positive integer $K$, and work with $K$-bounded clocks and clock
constraints.
\begin{definition}[Probabilistic Timed Game Arena (PTGA)]
  A probabilistic timed game arena is a tuple 
  $\pta {=} (L, \clocks, \inv, \act_\mMIN, \act_\mMAX, E, \delta)$
  where 
  \begin{itemize}
  \item
    $L$ is a finite set of \emph{locations};
  \item
    $\clocks$ is a finite set of \emph{clocks}; 
  \item
    $\inv : L {\to} \zones$ is an \emph{invariant condition}; 
  \item
    $\act_\mMIN$ and $\act_\mMAX$ are disjoint finite sets of \emph{actions}, and we use $\act$ for the set $\act_\mMIN\cup \act_\mMAX$
  \item 
    $E : L {\times}\act {\to} \zones$ is an \emph{action enabling condition}; 
  \item
    $\delta : L {\times} \act {\to} \DIST(2^{\clocks} {\times} L)$ is a
    \emph{probabilistic transition function}.  
  \end{itemize}
\end{definition}
When we consider a PTGA as an input of an algorithm, its size is understood as
the sum of the sizes of encodings of $L$, $\clocks$, $\inv$, $\act$, $E$, and
$\delta$. 
As usual~\cite{JLS08}, we assume that probabilities are expressed as ratios of
two natural numbers, each written in binary, and zones in the
definition of $\inv$ and $E$ are expressed as clock constraints.

A standard \emph{probabilistic timed automaton} (PTA) is a PTGA where one of $\act_\mMIN$ and $\act_\mMAX$ is empty. On the other hand, the standard (non-probabilistic) \emph{timed game arena} (\emph{timed automaton}) is a PTGA (PTA) such that $\delta(\ell,a)$ is a point distribution for all $\ell \in L$ and $a \in \act$.  

\subsection{Probabilistic Timed Game Arena: Semantics}

\noindent
Let $\pta{=}(L, \clocks, \inv, \act_\mMIN, \act_\mMAX, E, \delta)$
be a probabilistic timed game arena.
A \emph{configuration} of a PTGA is a
pair $(\ell, \nu)$, where $\ell$ is a location and $\nu$ a clock valuation such that 
$\nu \sat \inv(\ell)$. 
For any $t \in \Rplus$, we let $(\ell,\nu){+}t$ equal the configuration
$(\ell,\nu{+}t)$.   
In a configuration $(\ell,\nu)$, a timed action (time-action pair)
$(t,a)$ is available if and only if the invariant condition $\inv(\ell)$ 
is continuously satisfied while $t$ time units elapse, and $a$ is enabled
(i.e.\ the enabling condition $E(\ell,a)$ is satisfied) after $t$ time units
have elapsed.  
Furthermore, if the timed action $(t,a)$ is performed, then the next
configuration is determined by the probabilistic transition relation $\delta$,
i.e.\ with probability $\delta[\ell,a](C,\ell')$ the clocks in $C$ are reset and we move to the location $\ell'$. 

A game on a PTGA starts in an {\em initial configuration} $(\ell,\nu)\in L {\times} V$ and $\mMin$ and
$\mMax$ construct an infinite play by repeatedly choosing available  timed actions 
$(t_a, a)\in \Rplus {\times} \act_\mMIN$ and $(t_b, b)\in \Rplus {\times}
\act_\mMAX$ proposing $\bot$ if no timed action is available.  
The player responsible for the move is $\mMin$ if the time delay of $\mMin$'s choice is
less than that of $\mMax$'s choice or $\mMax$ chooses $\bot$, and otherwise $\mMax$ is
responsible. 
We assume the players cannot  simultaneously choose $\bot$, i.e.\ that in any configuration 
there is at least one timed action available.
\begin{definition}[PTGA Semantics]\label{ptgsem-def}
  Let $\pta=(L, \clocks, \inv, \act_\mMIN, \act_\mMAX, E, \delta)$
  be a PTGA.   
  The semantics of $\pta$ is given by the SGA
\[
\sem{\pta} {=} (S, A_\mMIN, A_\mMAX,  p_\mMIN, p_\mMAX, \winner, \tau_\mMIN, \tau_\mMAX)
\] 
where 
  \begin{itemize}
  \item 
    $S \subseteq L {\times} V$ is the (possibly uncountable) set of states such
    that $(\ell,\nu) \in S$ if and only if $\nu \sat \inv(\ell)$;  
  \item 
    $A_\mMIN = (\Rplus {\times} \act_\mMIN) \cup \set{\bot}$ and
    $A_\mMAX = (\Rplus {\times} \act_\mMAX) \cup \set{\bot}$ are the sets of
    timed actions of players $\mMin$ and $\mMax$; 
  \item
    for $\star\in\{\mMIN, \mMAX\}$, $(\ell,\nu) \in S$ and $(t,a) \in A_\star$
    the probabilistic transition function $p_\star$ is defined when
    $\nu {+} t' \sat \inv(\ell)$ for all $0 {\leq} t' {\leq} t$,
    ${\nu {+} t \sat E(\ell,a)}$ and for any $(\ell',\nu')$: 
    \[
    \begin{array}{c}
      p_\star( (\ell,\nu), (t, a))((\ell',\nu')) = \sum\nolimits_{C \subseteq \clocks \wedge
        (\nu{+}t)_C=\nu'} \delta[\ell,a](C,\ell');
    \end{array}  
    \]
  \item
    for $(t_a,a) \in\Rplus {\times} \act_\mMIN$ and $(t_b,b) \in\Rplus {\times} \act_\mMAX$, we define
\[
\winner((t_a,a),(t_b,b)) = \left\{ \begin{array}{cl} 
(t_a,a) & \mbox{if $t_a {<} t_b$} \\
(t_b,b) & \mbox{otherwise.}
\end{array} \right.
\]
If one of the arguments to $\winner$ is $\bot$, we define the returning value to be the other argument.
\item
    the time delay function is given by $\tau_\star(s,(t,a)) =
    t$ for all $\star\in\{\mMIN, \mMAX\}$, $s \in S$ and $(t,a) \in A_\star$
    such that $p_\star(s, (t, a))$ is defined. 
\end{itemize}
\end{definition}
The sum in the definitions of $p_\mMIN$ and $p_\mMAX$ is
used to capture the fact that resetting different subsets of $\clocks$
may result in the same clock valuation (e.g.\ if all clocks are initially zero,
then we end up with the same valuation, no matter which clocks we reset).
Also, notice that the time delay function of the SGA corresponds to the elapsed
time of each move.  

\paragraph*{Time Divergence} When modelling real-time systems it is important to restrict attention to time divergent (or non-Zeno) behaviour. More precisely, one should not consider strategies which lead to behaviour in which time does not advance beyond a certain point, as this cannot occur in a real system. We achieve this by restricting attention to \emph{structurally non-Zeno} PGTAs, these are PGTA where all strategies will yield time-divergent behaviour by construction. We use the syntactic conditions given in~\cite{NPS13} for PTAs and are derived from those for timed automata~\cite{Tri99,TYB05}.
\begin{figure}[t]  \centering
  {\small
\begin{tikzpicture}[node distance=2.3cm]
\tikzstyle{nloc}=[draw,circle,minimum size=1.4em,inner sep=0em]
\tikzstyle{probloc}=[draw,circle,fill=black,minimum size=.4em,inner sep=0em]
\tikzstyle{trans}=[-latex, rounded corners]
 \node[nloc, label={left:$x{,}y{\le}2$}] at (0,0) (l0)  {$\ell_0$};
 \node[probloc] at (2.5,0) (p0) {};
 \node[nloc, label={below:$(0{<}y{\le}2){\wedge}(x{\le}2)$}] at (5,0) (l1) {$\ell_1$};
 \node[probloc] at (7.5,0) (p1) {};
 \node[nloc, label={right:$x{,}y{\ge}2$}] at (9.5,0) (l2) {$\ell_2$};
 \node[probloc] at (5,1.6) (p3) {};
 \node[probloc] at (5,-1.9) (p4) {};

 \draw (p3)+(0.3,0) arc (0:-90:0.3);
 \draw (p0)+(0,-0.3) arc (-90:-180:0.3);

 \draw[rounded corners] (l0) |- (p3) node [pos=0.7, above] {$b, x{>}1$};
 \draw[trans] (p3) -| (l2) node[pos=0.25,above] {$0.5, x{:=}0, y{:=}0$};
 \draw[trans] (p3) -- (l1) node [pos=0.4, left] {$0.5,\,x{:=}0$};
 \draw (l1) -- (p1) node [midway, above] {$a, \; y{>}1$};
 \draw[trans] (p1) -- (l2) node[pos=0.5,above] {$x{:=}0, y{:=}0$};
 \draw[dashed] (l1) -- (p0) node [midway, above] {$c, \; y{>}1$};
 \draw[trans, dashed] (p0) -- (l0) node[pos=0.5,above] {$0.2,y{:=}0$};
 \draw[trans, dashed] (p0) |- +(1,-1.3) -| (l2.-110) node[pos=0.35,above] {$0.8,x{:=}0, y{:=}0$};
 \draw[rounded corners] (l0) |- (p4) node [pos=0.7, below] {$a, \; x{=}2$};
 \draw[trans] (p4) -| (l2.-70) node [pos=0.35, below] {$x{:=}0, y{:=}0$};
 \draw[trans] (l2) -- +(0.7,0.3) -- +(0.3,0.7) -- (l2) node[pos=-0.1,right] {$x{:=}0$};
\end{tikzpicture}
}
  \vspace*{-0.4cm}
  \caption{Example of a probabilistic timed game arena.}\label{fig:example}
  \vspace*{-0.2cm}
\end{figure}
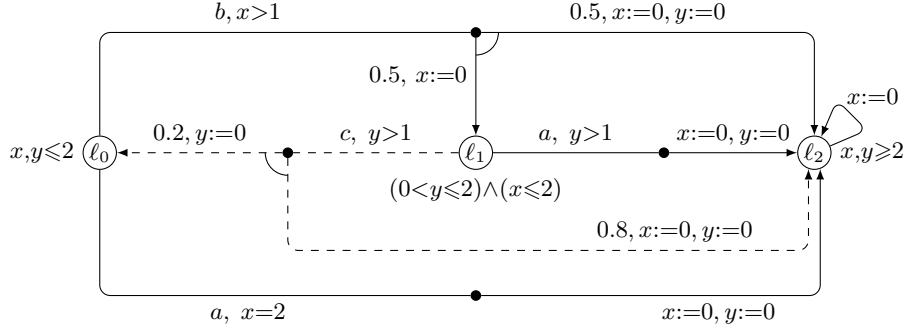
\begin{example}\label{ptga-example}
Consider the PTGA in Figure~\ref{fig:example}; we use solid
and dashed lines to indicate actions controlled by  $\mMin$
and $\mMax$ respectively.
Considering location $\ell_1$, the invariant condition is 
$(0{<}y{\le}2) {\wedge} (x{\le}2)$,  actions $a$ and $c$ are enabled when $y{>}1$
and, if $a$ is taken, we move to $\ell_2$, while if $c$ is taken, with
probability $0.2$ we move to $\ell_0$ and reset $y$, and with probability
$0.8$ move to $\ell_2$. 

Let us denote clock valuations by tuples where the first (second) coordinate
correspond to the clock $x$ ($y$).
Starting in the configuration  $(\ell_0, (0,0))$
and supposing $\mMin$'s strategy is to choose $(1.1,b)$ (i.e., wait $1.1$ time units
before performing action $b$) in location $\ell_0$ and then choose $(0.5,a)$ in location
$\ell_1$, while $\mMax$'s strategy in location $\ell_1$ is to choose $(0.2,c)$,
one possible play under this strategy pair is
\begin{multline*}
  \langle(\ell_0{,}(0{,}0)),\ ((1.1{,}b),\bot),\ (\ell_1{,}(0{,}1.1)),\\ 
  \linebreak[3] ((0.5{,} a),\ (0.2{,}c)),\ 
  (\ell_0{,} (0.2{,}0)),\ ((1.1{,} b),\bot),\ (\ell_2, (0{,}0))\rangle
\end{multline*}
which has probability $0.5 {\cdot} 0.2 {\cdot} 0.5 = 0.05$ and time
$1.1{+}0.2{+}1.1 = 2.4$ of reaching the location $\ell_2$. 
\end{example}

\subsection{Reachability-time problem over PTGA}

\noindent
We are interested in the reachability-time
problem for games over the semantics of a PTGA $\pta$. We
assume that the target set is given as a set $L_F$ of locations (the
corresponding target of the SGA $\sem{\pta}$, with state space $S$, is given by 
$F{=}\{ (\ell,\nu) \in S \mid \ell \in L_F \}$). However, the results
presented can be easily generalised to target sets of location-zone pairs.

\subsection{Non-determinacy of PTGA with reachability-time objectives}\label{subsection:det-ertg}

\noindent
Before proceeding with the definitions that we need to prove the
main decidability result of the paper, we show, through the following counter-example, that PTGAs are not determined, even when the game contains only non-strict inequalities.
\begin{figure}[t]
  \centering
{\small
\begin{tikzpicture}[node distance=2.3cm]
\tikzstyle{nloc}=[draw,circle,minimum size=1.4em,inner sep=0em]
\tikzstyle{probloc}=[draw,circle,fill=black,minimum size=.4em,inner sep=0em]
\tikzstyle{trans}=[-latex, rounded corners]
 \node[nloc, label={left:$x{\le}1$}] at (0,0) (l0)  {$\ell_0$};
 \node[probloc] at (2,1.5) (p0) {};
 \node[probloc] at (2,-1) (p1) {};
 \node[nloc, label={above:$x{\le}1$}] at (4,1.5) (l1) {$\ell_1$};
 \node[nloc, label={below:$x{\le}1$}] at (4,-1) (l1a) {$\ell_2$};
 \node[nloc, label={below:$x{\le}1$}] at (5,-0) (l1b) {$\ell_3$};
 \node[probloc] at (6,1.5) (p2) {};
 \node[nloc, fill=gray!20!white] at (8,0) (l2) {$\ell_4$};
 
 \node[] at (8.9,0.2) {$x{\geq}1$};
 \node[] at (8.95,-0.2) {$x{:=}0$};
 
 \draw[trans] (l0) |- (p1)  node [pos=0.75, below] {$x {\le} 1$} -- (l1a);
 \draw[trans, dashed] (l0) |- (p0)  node [pos=0.75, above] {$x {\le} 1$} -- (l1);
 \draw[trans, dashed] (l1) -- (p2) node [pos=0.5, above] {$x {=} 1$} -| (l2);
 \draw[trans] (l2) -- +(0.5,0.5) -- +(0.5,-0.5) -- (l2);
 \draw[trans,dashed] (l1a) |- node[pos=0.25,left] {$x{=}0$} node[pos=0.4,probloc] {}(l1b);
 \draw[trans] (l1a) -| node[pos=0.15,below] {$x {\le} 1$} node[pos=0.25,probloc] {}(l2);
 \draw[trans,dashed] (l1b) -- node[pos=0.25,below] {$x {=} 1$} node[pos=0.6,probloc] {}(l2);
 \draw[trans] (l1) -- node[pos=0.35,left] {$x {=} 0\phantom{xx}$} node[pos=0.5,probloc] {}(l2);
\end{tikzpicture}
}
  \caption{Example demonstrating that PTGAs are not determined. Action names are omitted for brevity.
    \label{fig:notdet}}  
\end{figure}
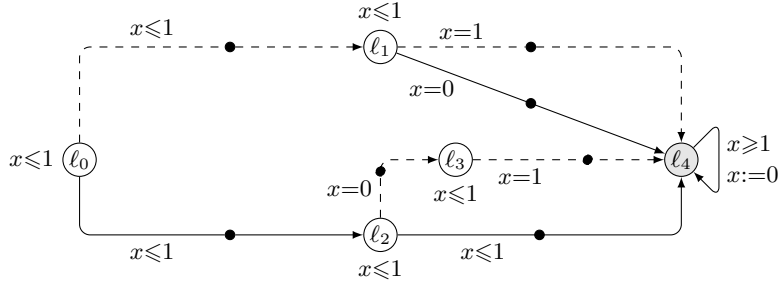
\begin{example}
Considering the PTGA given in Figure~\ref{fig:notdet} with target set $L_F {=} \{ \ell_4
\}$, recall that we use solid and dashed lines to indicate actions controlled
by $\mMIN$ and $\mMAX$ respectively. Constructing the optimality equations $\UOpt_\gam$ for the SGA
semantics of this PTGA, we have, after some simplifications:
\begin{align*}
P(\ell_4, x) = & \; 0 \\
P(\ell_3, x) = & \; 1{-}x \\
P(\ell_2, x) = & \; \left\{ \begin{array}{cl}
1 & \mbox{if $x{=}0$} \\
0 & \mbox{otherwise}
\end{array} \right. \\
P(\ell_1,x) = & \; \left\{ \begin{array}{cl}
0 & \mbox{if $x{=}0$} \\
1{-}x & \mbox{otherwise}
\end{array} \right.
\end{align*}
and $P(\ell_0, 0)$ is equal to the minimum of:
 \begin{equation}\label{nondet1-eqn}
\max\{0{+}P(\ell_2,0),0{+}P(\ell_1,0)\}
\end{equation}
and
\begin{equation}\label{nondet2-eqn}
\inf_{0 < t \le 1} \left\{ \max \left\{\sup_{t<t'\le 1} \big( t {+} P(\ell_2,t) \big), \sup_{0\le t' \leq t} \big( t' {+} P(\ell_1,t') \big) \right\}\right\} \, .
\end{equation}
The expression (\ref{nondet1-eqn}) is equal to $1$ and corresponds to player $\mMin$ leaving $\ell_0$ immediately (when the clock $x$ equals 0). The expression (\ref{nondet2-eqn}) corresponds to the infimum over leaving $\ell_0$ after a non-zero delay (when the clock $x$ is greater than 0) and is also equal to $1$. Combining these results we have that $P(\ell_0, 0){=}1$.

On the other hand, considering the optimality equations $\LOpt_\gam$, the values for the locations $\ell_1,\dots,\ell_4$ are as above, while the value for $P(\ell_0, 0)$ equals the maximum of:
\begin{equation}\label{nondet3-eqn}
\!\!0{+}P(\ell_1,0)
\;\;\mbox{and}
\;\;
\sup_{0 < t \le 1} \left\{ \min \left\{\inf_{t \leq t' {\leq} 1} \big( t {+} P(\ell_1, t)\big), \inf_{0 \le t' < t} \big( t' {+} P(\ell_2, t') \big) \right\}  \right\}\, .
\end{equation}
The first expression in (\ref{nondet3-eqn}) equals $0$ and corresponds to player $\mMax$ leaving $\ell_0$ immediately. The second expression in (\ref{nondet3-eqn}) corresponds to the supremum over leaving $\ell_0$ after a non-zero delay,
and is also equal to $0$, and therefore it follows that $P(\ell_0, \mathbf{0}){=}0$. Hence the game is not determined as the upper and lower values of the game differ in the state $(\ell_0,0)$.
\end{example}

\section{Boundary region abstraction}\label{subsec:boundary-region-graph}

\noindent
The region graph~\cite{AD94} is useful for solving time-abstract optimisation
problems on timed automata.   
The region graph, however, is not suitable for solving competitive optimisation
problems and games on timed automata as it abstracts away the timing
information.   
The corner-point abstraction~\cite{BBL04}, which captures digital clock semantics~\cite{HMP92}
of timed automata, is an abstraction of timed automata where the configurations
of the system are restricted to  $L {\times} \Nat^{|\clocks|}$, i.e.\ transitions
are allowed only when all clocks have non-negative integer values.
Although this abstraction retains some timing information, it is not convenient
for proof techniques based on dynamic programming, used in this paper.   
The boundary region abstraction (BRA)~\cite{JT07}, a generalisation of
the corner-point abstraction, is better suited for such proof techniques. 
More precisely, we need to prove certain properties of values in a PTGA, which we can do only when reasoning about all the states of the PTGA.
In the corner-point
abstraction we cannot do this since it represents only states corresponding to
corner points of regions.  
Here, we generalise the BRA of~\cite{JT07} to handle PTGAs. First, we require a
number of preliminary concepts. 
\paragraph{Timed Successor Regions}
Recall that $\Rr$ is the set of clock regions.
For $\region, \region' \in \Rr$, we say that $\region'$ is in the future of
$\region$, denoted $\region \rightarrow^{*} \region'$, 
if there exist $\nu \in \region$, $\nu' \in \region'$ and $t \in \Rplus$ such
that $\nu' = \nu{+}t$ and say $\region'$ is the \emph{time successor} of
$\region$ if $\region \neq \region'$ and $\nu{+}t' \in \region \cup \region'$ for all $t' {\leq} t$  and write 
$\region \rightarrow \region'$ to denote this fact.
We also use $\region \rightarrow^+ \region'$ if there is $\region''$
such that $\region \rightarrow \region'' \rightarrow^* \region'$.
For regions $\region, \region' \in \Rr$ such that $\region \rightarrow^*
\region'$ we write $[\region, \region']$ for the zone $\cup \{ \region''
  \mid  \region \rightarrow^* \region'' \wedge \region''
  \rightarrow^* \region' \}$.
\paragraph{Intuition for the Boundary Region Abstraction}
In our definition of the boundary region abstraction (BRA) we
capture the intuition that, when studying the ``optimal'' behaviour of the players,
it is sufficient to consider moves that take place near the start or end of
the regions. This allows us to abstract from moves that specify the precise time,
but instead allow the players to say which regions they wish to enter, and then
either say that they want to take the move at the start of the region ($\inf$), or
at its end ($\sup$). 

Based on this intuition we define the boundary region abstraction of
a probabilistic game arena as follows.
\begin{definition}[Boundary region abstraction (BRA)]\label{brg-def}
  For a probabilistic timed game arena
  $\pta{=}(L, \clocks, \inv, \act_\mMIN, \act_\mMAX, E, \delta)$,
  the boundary region abstraction of $\pta$ is given by
  the SGA
  $
  \RegB{\pta} {=} (\hS, \hA_\mMIN, \hA_\mMAX, \hp_\mMIN, \hp_\mMAX, \hwinner, \htau_\mMIN, \htau_\mMAX)
  $
  where  
  \begin{itemize}
  \item
    $\hS \subseteq L {\times} V {\times} \Rr$ is the (possibly uncountable) set of
    states such that $(\ell, \nu, \region) \in \hS$ if and only if 
    $\region \subseteq \inv(\ell)$ and $\nu \in \CLOS{\region}$ (recall that $\CLOS{\region}$ denotes the closure of $\region$);
  \item
    $\hA_\mMIN = (\act_\mMIN {\times} \Rr {\times} \{\inf,\sup\})
    \cup \set{\bot}$ is the set of actions of player $\mMIN$;
  \item
    $\hA_\mMAX = (\act_\mMAX {\times} \Rr {\times} \{\inf,\sup\})
    \cup \set{\bot}$ is the set of actions of player $\mMAX$;
  \item 
  	for $\star \in \{ \mMin , \mMax \}$,  
    $\hs = (\ell, \nu, \region) \in \hS$ and $\alpha = (a,
      \region'', \oper) \in \hA_\star$ such that $\region \rightarrow^* \zeta''$, the
      probabilistic transition function $\hp_\star$ is
      defined if  $[\region, \region''] \subseteq \inv(\ell)$ and
      ${\region'' \subseteq E(\ell, a)}$ and for any $(\ell',\nu',\region') \in \hS$:
    \[ \begin{array}{c}
    \hp_\star( \hs , \alpha ) ((\ell',\nu',\region')) =
    \sum\nolimits_{C \subseteq \clocks \wedge \nu''_C=\nu' \wedge  
      \region''_C=\region'} \delta[\ell,a](C,\ell')
    \end{array} \]
    where $\nu'' = \oper_{\nu + t\in \region'', t\ge 0} \nu {+} t$;
  \item
    $\hwinner((a,\region_a,\oper_a),(b,\region_b,\oper_b))$ is equal to $(a,\region_a,\oper_a)$ if (i)
    $\region_a \rightarrow^+ \region_b$ or (ii) $\region_a=\region_b$, $\oper_a = \inf$ and $\oper_b=\sup$; it is 
    equal to $(b,\region_b,\oper_b)$ otherwise;
  \item
    for $\star \in \{ \mMin , \mMax \}$,  $(\ell, \nu, \region) \in \hS$ and $(a_\alpha, \region_\alpha,\oper) \in \hA_\star$
    such that $\hp_\star$ is defined the time delay function is given by 
    $\htau_\star((\ell, \nu, \region),(a_\alpha, \region_\alpha,\oper)) = \oper_{\nu + t\in \region_\alpha} t$.
  \end{itemize} 
Given a target set of locations $L_F$ of $\pta$, the corresponding target set of the BRA is given by $\hF  {=} \{ (\ell,\nu,\region) \in \hS \mid  \ell \in L_F \}$.
\end{definition}
To simplify notation, for two elements $a \in \hA_\mMIN$ and $b\in \hA_\mMAX$ we write $a {\le} b$ to denote that $\hwinner(a,b) {=} a$. We use analogous notation also for other SGAs.
For an element $s {=} (\ell,\nu) \in L {\times} V$, we use $\RegB{\smash{s}}$ to denote the element $(\ell,\nu,[\nu])\in \hS$.

Although the boundary region abstraction is not a finite SGA,
for a fixed initial state we can restrict attention to a finite SGA, adapting an approach from~\cite{Tri09} as follows.
\begin{proposition}
  \label{prop:reachable-subgraph-is-finite-game}
  Let $\pta$ be a PTGA and $\RegB{\pta}$ the corresponding BRA. 
  For any state of $\RegB{\pta}$, its reachable sub-graph is
  finite and constructible in time exponential in the size of $\pta$.
\end{proposition}
\begin{proof}
  The most demanding part of the proof is to show that there is a set $V$ of valuations that has exponential size and contains $\nu$ for
  any state $(\ell, \nu, \region)$ reachable in the sub-graph of $\RegB{\pta}$.

  For $r \in \Rplus$ we write $\FRAC{r}$  for the
  fractional part of $r$, i.e. $r {-}\FLOOR{r}$. 
  For a clock valuation $\nu$ we define its fractional signature~$\FSIG{\nu}$
  to be the sequence $(f_0, f_1, \dots, f_m)$ such that $f_0 {=} 0$, $f_i {<} f_j$
  if $i {<}  j$, for all $i, j \leq m$, and $f_1, f_2, \dots, f_m$ are all the
  non-zero fractional parts of clock values in the clock valuation $\nu$.
  In other words, for every $i {\geq} 1$ there is a clock $c$ such
  that  $\FRAC{\nu(c)} {=} f_i$, and for every clock $c \in \clocks$
  there is $i {\leq} m$ such that $\FRAC{\nu(c)} {=} f_i$.

  Let $\oplus$ denote addition modulo $m$.
  For $0 {\leq} k {\leq} m$ we define the $k$-shift of a
  fractional signature $(f_0, f_1, \ldots, f_m)$ as the fractional
  signature $(f'_0, f'_1, \ldots, f'_m)$
  such that for all $0\le i \leq m$ we have
  $f'_i = \FRAC{f_{i {\oplus} k} + 1 {-} f_k}$.  
  Note that a $k''$-shift
  $(f''_{0}, \ldots, f''_{m})$
  of a $k'$-shift
  $(f'_{0}, \ldots, f'_{m})$
  of
  $(f_0, f_1, \ldots, f_m)$
  is an $(k'\oplus k'')$-shift of 
  $(f_0, f_1, \ldots, f_m)$
  because for any $i$ we have:
  \begin{align*}
  f''_{i}
  =& \; \FRAC{f'_{i\oplus k''} {+} 1 {-} f'_{k''}} \\
  =& \; \FRAC{\FRAC{f_{i\oplus k'' \oplus k'} {+} 1{-}f_{k'}} {+} 1 {-} \FRAC{f_{k''\oplus k'} {+} 1{-}f_{k'}}} \\
  =& \;  \FRAC{f_{i\oplus (k''\oplus k')} {+} 1 {-} f_{k'' \oplus k'}} \,.
  \end{align*}
  This means that, by successive application of shifts, only $m$ different fractional signatures can be obtained. We further say that a fractional signature $(f_0', f_1', \ldots, f_n')$ is a
  subsequence of another fractional signature
  $(f_0, f_1, \ldots, f_m)$ if $n {\leq} m$ and for all $i {\le} n$
  there exists $j {\leq} m$ such that $f_i' {=} f_j$.

  For any state $(\ell, \nu, \region)$ of the BRA $\RegB{\pta}$, we claim that it is only possible to transition to states
  $(\ell',\nu',\region')$ such that $\FRAC{\nu'}$ is a subsequence of a $k$-shift of $\FRAC{\nu}$, for some $k$.
  To see that, notice that the $\nu_\alpha$ in the definition of $\hp_\star$ (Definition~\ref{brg-def})
  satisfies that $\FRAC{\nu_\alpha}$ is a $k$-shift of $\FRAC{\nu}=(f_0,\ldots f_m)$ for $k$ chosen so that $f_m$ is the fractional
  part of clocks that have integer value in $\nu_\alpha$. Subsequently resetting clocks gives rise to a subsequence of a fractional signature,
  and so $\FRAC{\nu'}$ (for $\nu'$ from the defining sum of $\hp_\star$) is a subsequence of $\FRAC{\nu_\alpha}$. \qed
\end{proof}
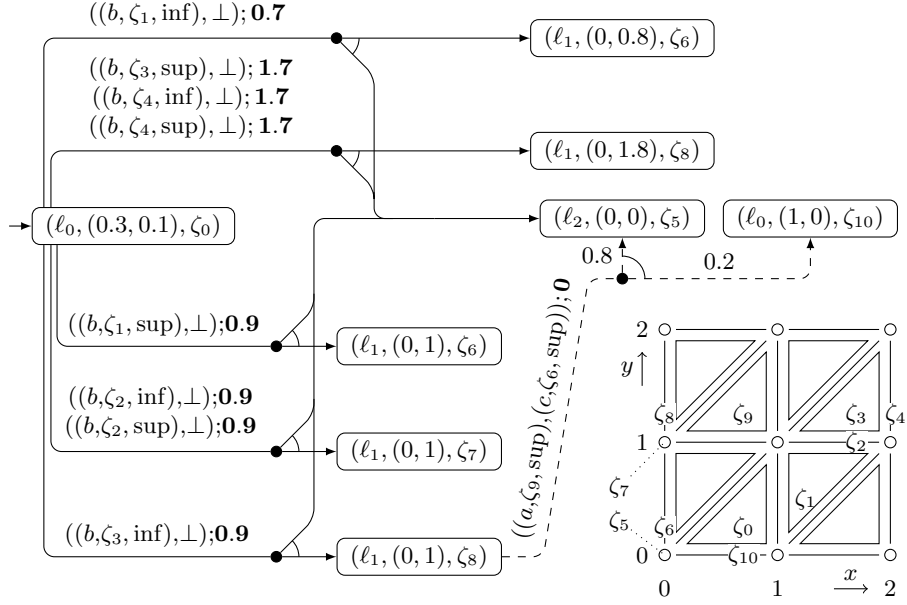
\begin{figure}[t]
  \centering
    {
\begin{tikzpicture}
\tikzstyle{every node}=[font=\small]
\tikzstyle{loc}=[rounded corners=3pt,draw,minimum size=1.4em,inner sep=0em]
\tikzstyle{probloc}=[draw,circle,fill=black,minimum size=.4em,inner sep=0em]
\tikzstyle{trans}=[-latex, rounded corners]
 \node[loc] at (0,0) (l0) {$\begin{array}{c}(\ell_0, (0.3,0.1),\region_0)\end{array}$};
 \node[loc] at (6.5,2.5) (l1) {$\begin{array}{c}(\ell_1, (0,0.8), \region_6)\end{array}$};
 \node[loc] at (6.5,1) (l2) {$\begin{array}{c}(\ell_1, (0,1.8), \region_8)\end{array}$};
 \node[loc] at (3.8,-1.6) (l3) {$\begin{array}{c}(\ell_1, (0,1), \region_6)\end{array}$};
 \node[loc] at (3.8,-3) (l4) {$\begin{array}{l}(\ell_1, (0,1), \region_7)\end{array}$};
 \node[loc] at (3.8,-4.4) (l5) {$\begin{array}{l}(\ell_1, (0,1),\region_8 )\end{array}$};

 \node[loc] at (9,0.1) (l6) {$\begin{array}{l}(\ell_0, (1,0),\region_{10})\end{array}$};

 \node[loc] at (6.5,0.1) (lf) {$\begin{array}{l}(\ell_2, (0,0),\region_5)\end{array}$};

 \node[probloc] at (2.7,2.5) (p1) {};
 \node[probloc] at (2.7,1) (p2) {};
 \node[probloc] at (1.9,-1.6) (p3) {};
 \node[probloc] at (1.9,-3) (p4) {};
 \node[probloc] at (1.9,-4.4) (p5) {};

 \node[probloc] at (6.5,-0.7) (p6) {};

 \draw[trans] ($(l0.-180) - (0.3,0)$) -- (l0);
 \draw[rounded corners] (l0.-192)|-(p1) node[pos=0.75,above]
  {$\begin{array}{c}((b,\region_1,\inf),\bot);\mathbf{0.7}\end{array}$};   
 \draw[rounded corners] (l0.-193)|-(p2) node[pos=0.75,above] {$\begin{array}{c}((b,\region_3,\sup),\bot);\mathbf{1.7}\\((b,\region_4,\inf),\bot);\mathbf{1.7}\\((b,\region_4,\sup),\bot);\mathbf{1.7}\end{array}$};
 \draw[rounded corners] (l0.194)|-(p3) node[pos=.75,above] {$((b{,}\region_1,\sup){,}{\bot}){;}\mathbf{0.9}$};
 \draw[rounded corners] (l0.193)|-(p4) node[pos=.75,above] {$\begin{array}{c}((b{,}\region_2,\inf){,}{\bot}){;}\mathbf{0.9}\\((b{,}\region_2,\sup){,}{\bot}){;}\mathbf{0.9}\end{array}$};
 \draw[rounded corners] (l0.192)|-(p5) node[pos=.75,above] {$((b{,}\region_3,\inf){,}{\bot}){;}\mathbf{0.9}$};

 \draw[rounded corners] (p1) -- +(0.5,-0.5) |- (4,0.1);
 \draw[rounded corners] (p2) -- +(0.5,-0.5) -- +(0.5,-0.55);
 \draw[rounded corners] (p3) -- +(0.5,0.5) -- +(0.5,0.7);
 \draw[rounded corners] (p4) -- +(0.5,0.5) -- +(0.5,0.7);
 \draw[rounded corners] (p5) -- +(0.5,0.5) |- (3.45,0.1);
 \draw[trans] (4,0.1) -- (lf);

 \draw[trans] (p1) -- (l1);
 \draw[trans] (p2) -- (l2);
 \draw[trans] (p3) -- (l3);
 \draw[trans] (p4) -- (l4);
 \draw[trans] (p5) -- (l5);
 
 \draw[trans, dashed] (p6) -- (lf) node[midway,left] {$0.8$};
 \draw[trans, dashed] (p6) -| (l6) node[pos=0.25,above] {$0.2$};

 \draw (p1)+(0.3,0) arc (0:-45:0.3);
 \draw (p2)+(0.3,0) arc (0:-45:0.3);
 \draw (p3)+(0.3,0) arc (0:45:0.3);
 \draw (p4)+(0.3,0) arc (0:45:0.3);
 \draw (p5)+(0.3,0) arc (0:45:0.3);
 \draw (p6)+(0,0.3) arc (90:0:0.3);
 

 \draw[dashed,rounded corners] (l5.0) -- ++(0.55,0) -- +(0.6,3.7) node [pos=0.52,sloped,above] {$\begin{array}{l}((a{,}\region_9,\sup){,}
(c{,}\region_6{,}\sup));\mathbf{0}\end{array}$} -- (p6);

\node (X) at (8.3,-3.1) {
\begin{tikzpicture}[x=1.5cm,y=1.5cm]
\foreach \i in {0,1}
{
 \foreach \j in {0,1}
 {
   \draw (\i+0.1,\j+0) -- (\i+0.9,\j+0);
   \draw (\i+0,\j+0.1) -- (\i+0,\j+0.9);
   \draw (\i,\j) circle (0.05);
   \draw (\i+0.2,\j+0.1) -- (\i+0.9,\j+0.1) -- (\i + 0.9, \j + 0.8) -- cycle;
   \draw (\i+0.1,\j+0.2) -- (\i+0.8,\j+0.9) -- (\i + 0.1, \j + 0.9) -- cycle;
   \draw (\i+0.1,\j+0.1) -- (\i+0.9,\j+0.9);
 }  
}
\foreach \i in {0,1}
{
 \foreach \j in {2}
 {
   \draw (\i+0.1,\j+0) -- (\i+0.9,\j+0);
   \draw (\i,\j) circle (0.05);
 }  
}
\foreach \i in {2}
{
 \foreach \j in {0,1}
 {
   \draw (\i+0,\j+0.1) -- (\i+0,\j+0.9);
   \draw (\i,\j) circle (0.05);
 }  
}
\draw (2,2) circle (0.05);

 \node[font=\small] at (0.7,0.25) {$\region_0$};
 \node[font=\small] at (1.25,0.5) {$\region_1$};
 \node[font=\small,fill=white,inner sep=0] at (1.7,1) {$\region_2$};
 \node[font=\small] at (1.7,1.25) {$\region_3$};
 \node[font=\small,fill=white,inner sep=0] at (2.05,1.25) {$\region_4$};
 \node[font=\small,fill=white,inner sep=0] (r5) at (-0.4,0.3) {$\region_5$};
 \draw[dotted] (r5) -- (0,0);
 \node[font=\small,fill=white,inner sep=0] at (0,0.25) {$\region_6$};
 \node[font=\small,fill=white,inner sep=0] (r7) at (-0.4,0.6) {$\region_7$};
 \draw[dotted] (r7) -- (0,1);
 \node[font=\small,fill=white,inner sep=0] at (0,1.25) {$\region_8$};
 \node[font=\small] at (0.7,1.25) {$\region_9$};
 \node[font=\small,fill=white,inner sep=0] at (0.7,0) {$\region_{10}$};

 \node at (-0.2,0) {$0$};
 \node at (-0,-0.3) {$0$};
 \node at (-0.2,1) {$1$};
 \node at (-0.2,2) {$2$};
 \node at (1,-0.3) {$1$};
 \node at (2,-0.3) {$2$};
\draw[->] (1.5,-0.3) -- (1.8,-0.3) node[midway,above] {$x$};
\draw[->] (-0.2,1.5) -- (-0.2,1.8) node[midway,left] {$y$};
\end{tikzpicture}};
\end{tikzpicture}
}
    \vspace*{-0.4cm}
    \caption{Sub-graph of the boundary region abstraction for the PTGA of
      Figure~\ref{fig:example}, with the region names as depicted in the bottom right corner.}\label{fig:finboundary}  
\end{figure}
\begin{example}\label{bra-example}
Returning to Example~\ref{ptga-example} (see Figure~\ref{fig:example}), a sub-graph of BRA reachable from $(\ell_0, (0.3,0.1), 0{<}y {<} x{<}1)$ for the PTGA of Figure~\ref{fig:example} is shown in Figure~\ref{fig:finboundary}. The names of the regions correspond to the regions depicted in the bottom right corner. Edges are labelled $(a, \region, \oper)$ and the intuitive meaning is to wait until we reach the lower or upper (depending on $\oper$) boundary of the region. For some regions, for example $\region_4$, the boundaries coincide and we keep this redundancy to simplify the notation. Considering the region $\region_1$, we see that it is determined by the constraints $(1{<}x{<}2) {\wedge} (0{<}y{<}1)  {\wedge} (y{<}x{-}1)$. The bold numbers on edges correspond to the time delay before the action labelling the edge is taken. Figure~\ref{fig:finboundary} includes the actions available in the initial state and one of the action pairs that are available in the state $(\ell_1, (0,1), (x{=}0){\wedge}(1{<}y{<}2))$. To simplify the figure, the probabilities that are equal to $0.5$ are omitted.
\end{example}

\section{Decidability of the Reachability-Time Problem}\label{sec:reach} 

\noindent
In this section we show decidability of the reachability-time
problem, which is the main result of the paper. The result
is formalised in the following theorem.
\begin{theorem}\label{theorem:main-result}
Let $\pta$ be a PTGA. The reachability-time problem for infinite-horizon objectives in $\pta$ is in NEXPTIME$\cap$co-NEXPTIME.
\end{theorem}
The crucial, and most demanding, step of the proof of Theorem~\ref{theorem:main-result} is proving
that the problems on PTGAs can be reduced to
problems on BRAs. This fact is formalised in Theorem~\ref{theorem:main-thm-red}.
Theorem~\ref{theorem:main-result} then follows straightforwardly from Theorem~\ref{theorem:main-thm-red},
Proposition~\ref{prop:reachable-subgraph-is-finite-game} and
Theorem~\ref{theorem:finite-games}.
\begin{theorem}\label{theorem:main-thm-red}
Let $\pta$ be a PTGA and $\RegB{\pta}$ the corresponding BRA. The answers to the reachability-time problems for $\pta$ and $\RegB{\pta}$ are the same.
\end{theorem}
The remainder of this paper is devoted to the proof of Theorem~\ref{theorem:main-thm-red}.
First, in Section~\ref{sect:qsfunct} we introduce quasi-simple functions and prove some of their properties. Then, in Section~\ref{sect:proof-core} we show that values in the games we study can be characterised using quasi-simple functions, and that this allows us to establish the correspondence between PTGA and its boundary region abstraction. 

For the remainder of this section, unless otherwise specified, we fix a PTGA $\pta= (L, \clocks, \inv, \act_\mMIN, \act_\mMAX, E, \delta)$, set of target locations $F_L$, suppose the semantics
of $\pta$ is given by:
\[
\sem{\pta} = (S,
A_\mMIN, A_\mMAX,  p_\mMIN, p_\mMAX, \winner, \tau_\mMIN, \tau_\mMAX)
\]
with corresponding target set $F {=} \{ (\ell,\nu) \in S \mid \ell \in F_L \}$
and the boundary region abstraction of $\pta$ is given by 
\[
\hT = (\hS, \hA_\mMIN,
\hA_\mMAX, \hp_\mMIN, \hp_\mMAX, \hwinner, \htau_\mMIN, \htau_\mMAX ) 
\]
with corresponding target set $\hF{=}\{ (\ell,\nu,\zeta) \in \hS \mid \ell \in F_L \}$.

\subsection{Quasi-simple Functions}\label{sect:qsfunct}

\noindent 
To prove properties of controllers for (non-probabilistic) timed systems, Asarin and Maler~\cite{AM99} introduced simple functions, a finitely
representable class of functions with the property that every decreasing
sequence is finite. We define these functions here and show that they are not sufficient for our purpose.
\begin{definition}[Simple Functions]
  Given a set of valuations $X {\subseteq} V$, a function $f : X {\to} \Rplus$ is \emph{simple} if
  there exists $e \in \Nat$  
  and either $f(\nu) {=} e$ for all $\nu \in X$, or there
  exists a clock $c \in C$ such that $f(\nu) {=} e {-} \nu(c)$ for all
  $\nu \in X$. Furthermore, a function $f : \hS {\to} \Rplus$ is regionally simple if $f(\ell, \cdot ,
\region)$ is simple for all $\ell  \in L$ and  $\region \in \Rr$.
\end{definition}
For timed games, Asarin and Maler showed that upper values for
$n$-step reachability-time objectives are regionally simple, and because the fixpoint is reached for some $n$
the upper value for reachability-time objective is regionally simple. 
Also, using the properties of simple functions,~\cite{JT07} shows that, for a
non-probabilistic game reachability-time objectives, the optimal strategies are
\emph{regionally positional}, i.e., in every state of a region the strategy
chooses the same action. Unfortunately, in the case of PTGAs, applying the value improvement function does not necessarily preserve regional-simplicity.  Moreover, as the example below demonstrates, neither is the value of the game necessarily regionally simple nor optimal strategies regionally positional.  
\begin{figure}[t]
  \centering
  {\small
\begin{tikzpicture}[node distance=2cm]
\tikzstyle{minloc}=[draw,circle,minimum size=1.4em,inner sep=0em]
\tikzstyle{maxloc}=[draw,minimum size=1.4em,inner sep=0em]
\tikzstyle{probloc}=[draw,circle,fill=black,minimum size=.4em,inner sep=0em]
\tikzstyle{trans}=[-latex, rounded corners]
 \node[minloc, label={left:$0 {\le} x {\le} 1$}] at (0,0) (l0) {$\ell_0$};
 \node[minloc] at (3,0) (l1) {$\ell_1$};
 \node[probloc] at (5,0) (p1) {};
 \node[minloc, fill=gray!20!white] at (6,0) (l2) {$\ell_2$};

 \node[] at (6.9,0.2) {$x{\geq}1$};
 \node[] at (6.95,-0.2) {$x{:=}0$};

 \node[probloc] at (1.5,0.8) (pa) {};
 \node[probloc] at (3,-0.8) (pb) {};

 \draw (pa)+(0.3,0) arc (0:-90:0.3);

 \draw[trans] (l0) |- (pa) node[pos=0.75,above] {$a$};
 \draw[trans] (l0) |- (pb) node[pos=0.65,below] {$b,\,x{=}1$};
 \draw[trans] (pa) -| (l2) node[pos=0.4,below] {$0.5$};
 \draw[trans] (pa) |- (l1) node[pos=0.75,above] {$x{:=}0$} node[pos=0.25,left] {$0.5$};
 \draw[trans] (pb) -| (l2);

 \draw[trans] (l1) -- (p1) node[midway,above] {$x{=}1$} ;
 \draw[trans] (p1) -- (l2);
 \draw[trans] (l2) -- +(0.5,0.5) -- +(0.5,-0.5) -- (l2);

\end{tikzpicture}
}
  \caption{Example demonstrating optimal strategies are not
    regionally positional.
    \label{fig:nonreg}}  
\end{figure}
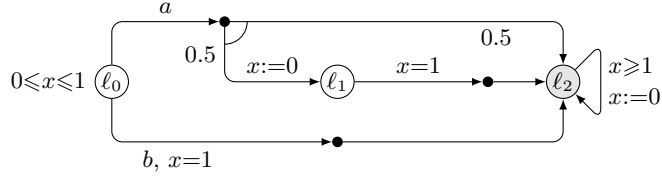
\begin{example} 
  Consider the one-player PTGA shown in
  Figure~\ref{fig:nonreg}.
  Observe that, for every state $(\ell_0,\nu)$ in the region 
  $(\ell_0, 0{<}x{<}1)$,
  the optimal expected time to reach $\ell_2$ equals  
  \[\min \left\{ \inf\nolimits_{t \geq 0} 
  \{ t + 0.5 {\cdot} 1 + 0.5 {\cdot} 0 \}, 1 {-} \nu(x) \right\} 
  = \min \{0.5, 1 {-} \nu(x) \}.
  \]
  Hence, the values in PTGA with reachability-time objectives may not be regionally simple. 
  Moreover, the optimal strategy is not regionally positional, since if
  $\nu(x) {\leq} 0.5$ then the optimal strategy is to take action $a$
  immediately, while otherwise the optimal strategy is to wait until
  $\nu(x) {=} 1$ and then take action $b$. 
\end{example}
Due to these results it is not possible to work with simple functions. Our proof instead relies on regional non-expansiveness of value functions. Given $X\subseteq V$, a function
$f:X {\to} \Rplus^\infty$ is {\em non-expansive} if for all $x,y\in X$ we have $\norm{f(x) {-} f(y)} \le \norm{x {-} y}$.
A function $f:\hS {\to} \Rplus^\infty$ is {\em regionally non-expansive} if $f(\ell, \cdot,\region)$ is non-expansive for
any $\ell$ and $\region$, and similarly any $f:S {\to} \Rplus^\infty$ is regionally non-expansive if
$f(\ell,\cdot)$ is non-expansive when its domain is restricted to a single region.

The proof direction that we take requires us to establish that
$\lim_{n\rightarrow\infty} \UVAL^n$ is non-expansive. To do this, we will show that for each $n \in \Nat$ the function $\UVAL^n$ is non-expansive. However, a direct proof by
induction would fail and instead we are required to prove a stronger claim about the functions $\UVAL^n$. To do this, we first 
introduce \emph{quasi-simple functions}.
\begin{definition}[Quasi-Simple Functions]\label{quasi-def}
  Let $X \subseteq V$ be a set of clock valuations. 
  The class of quasi-simple functions is built by first defining
  every simple function to be quasi-simple, and then inductively by stipulating
  that convex combination, maximum and minimum of finitely many quasi-simple functions
  are quasi-simple.
\end{definition}
A function $f : \hS {\to} \Rplus^\infty$ is regionally quasi-simple if $f(\ell, \cdot ,
\region)$ is quasi-simple for all $\ell  \in L$ and  $\region \in \Rr$,
and any $f:S {\to} \Rplus^\infty$ is regionally quasi-simple if
$f(\ell,\cdot)$ is quasi-simple when its domain is restricted to a single region.

We will later show that functions $\UVAL^n_{\sem{\pta}}:S {\to} \Rplus$ and $\UVAL^n_{\hT}:\hS {\to} \Rplus$ for $n \in \Nat$ are regionally quasi-simple. From this using the lemma below we can then demonstrate that these functions are non-expansive.
\begin{lemma}\label{lem:non-expansive}
Every quasi-simple function is non-expansive.
\end{lemma}
\begin{proof}
Consider any quasi-simple function $f : X {\to} \Rplus^\infty$.
We will prove by induction on the structure of $f$ (see Definition~\ref{quasi-def}) that for any $\nu_1,\nu_2 \in X$ we have $|f(\nu_1) - f(\nu_2)| \le |\nu_1-\nu_2|$.
\begin{itemize}
\item
If $f$ is a simple function, then either $f$ is a constant, and hence:
\[
\norm{f(\nu_1) - f(\nu_2)} = 0 \le |\nu_1-\nu_2| 
\]
or $f = e - \nu(c)$ for some clock $c$, in which case:
\[
\norm{f(\nu_1) - f(\nu_2)} = \norm{\nu_2(c) - \nu_1(c)} \le \norm{\nu_2 - \nu_1} \,
\]
as required.
\item
If $f$ is a convex combination $p_1,\ldots ,p_n$ of quasi-simple functions $f_1,\ldots, f_n$, then:
\begin{align*}
\lefteqn{\hspace{-1cm} \norm{ \mbox{$\sum_{i=1}^n$} p_i {\cdot} f_i(\nu_1) - \mbox{$\sum_{i=1}^n$} p_i {\cdot} f_i(\nu_2)}  \le 
\norm{\mbox{$\sum_{i=1}^n$} p_i {\cdot} (f_i(\nu_1) - f_i(\nu_2))}} \\
\le & \; \mbox{$\sum_{i=1}^n$} p_i {\cdot} \norm{\nu_1 -\nu_2} \tag{by induction} \\
= & \; \norm{\nu_1 -\nu_2} \tag{since we are considering a convex combination}
\end{align*}
as required.
\item
If $f$ is the maximum of two quasi-simple functions $f_1$ and $f_2$, then without loss of generality we suppose $f_1(\nu_1) \ge f_2(\nu_1)$. In the case when
$f_1(\nu_2) \ge f_2(\nu_2)$ we have
\[
\norm{\max\{f_1(\nu_1),f_2(\nu_1)\} - \max\{f_1(\nu_2),f_2(\nu_2)\}} \le
\norm{f_1(\nu_1) - f_1(\nu_2)} \le \norm{\nu_1 - \nu_2}
\]
since $f_1$ is non-expansive. On the other hand, in the case when $f_1(\nu_2) < f_2(\nu_2)$:
\[
\norm{\max\{f_1(\nu_1),f_2(\nu_1)\} - \max\{f_1(\nu_2),f_2(\nu_2)\}} \le
\norm{f_1(\nu_1) - f_2(\nu_2)} \, .
\]
Now either $f_1(\nu_1) \ge f_2(\nu_2)$, and therefore we have:
\[
\norm{f_1(\nu_1) - f_2(\nu_2)}\le \norm{f_1(\nu_1) - f_1(\nu_2)}
\le \norm{\nu_1 - \nu_2}
\]
since $f_1$ is non-expansive, or $f_1(\nu_1) < f_2(\nu_2)$ in which case:
\[
\norm{f_1(\nu_1) - f_2(\nu_2)} \le \norm{f_2(\nu_1) - f_2(\nu_2)}
\le \norm{\nu_1 - \nu_2}
\]
since $f_2$ is non-expansive. Since these are all the cases to consider we have $f$ is non-expansive as required.
\item
If $f$ is minimum of two quasi-simple functions the proof follows similarly to the case when $f$ is the maximum of two quasi-simple functions.
\end{itemize}
Since these are the only cases to consider the proof is complete. \qed
\end{proof}
In the proofs below we will make use of several technical properties of quasi-simple functions. First, however, we require an alternative representation of quasi-simple functions in terms of parse trees.

Let $\Upsilon$ be the set of all parse trees whose 
leaves are simple functions and whose nodes are the operations: min, max and convex combination. Clearly, every tree $\Delta \in \Upsilon$ corresponds to a unique quasi-simple function which we will call $\mathit{qs}(\Delta)$. Conversely, every quasi-simple function corresponds to infinitely many trees from $\Upsilon$. The definition below gives us a unique representative.
\begin{definition}\label{def:representative-tree}
Let the rank of a quasi-simple function $f : X {\to} \Rplus^\infty$, denoted $\rank{f}$, be
the smallest $k$ such that there is a tree $\Delta\in \Upsilon$ of height $k$ such that $\mathit{qs}(\Delta) = f$.
For any quasi-simple function $f : X {\to} \Rplus^\infty$ we define a unique representative parse tree $\Delta_f$ by induction on the rank of $f$.
\begin{itemize}
\item
If $\rank{f}=0$, then let $\Delta_f$ to be any tree with height $0$ such that $\mathit{qs}(\Delta_f) = f$.
\item
If $\rank{f}=k{+}1$ for some $k \in \Nat$, there must be an operation $\mathit{op}$ (either min, max or convex combination) and integer $n$ such that $f$ is obtained by taking the $\mathit{op}$ of the quasi-simple functions $f_1,\ldots, f_{n}$, each of which has rank at most $k$. Therefore, by induction we have representatives $\Delta_{f_1},\ldots, \Delta_{f_n}$ for $f_1,\dots,f_n$. Let $\Delta_f$ be the tree with root $\mathit{op}$ and subtrees $\Delta_{f_1},\ldots, \Delta_{f_n}$. Clearly, by construction we have $\mathit{qs}(\Delta_f) = f$.
\end{itemize} 
\end{definition}
The first technical property of quasi-simple functions will allow us to establish that when we take a delay so that a boundary of a region is reached, quasi-simplicity is preserved.
\begin{lemma}\label{lemma:quasi-next-region}
Let $f : X {\to} \Rplus^\infty$ be a quasi-simple function, $c$ a clock and $i$ an integer such that $\nu(c) {\ge} i$ for all $\nu\in X$.
If $\felapse{f}  : X {\to} \Rplus^\infty$ is the function where for any $\nu \in X$ we have $\felapse{f}(\nu) = t_\nu {+} f(\nu {+}t_\nu)$ and $t_\nu = \nu(c) {-} i$, then $\felapse{f}$ is quasi-simple.
\end{lemma}
\begin{proof}
Consider any quasi-simple function $f : X {\to} \Rplus^\infty$ and let $\Delta_f$ be its representative parse tree constructed using Definition~\ref{def:representative-tree}. Let $\Delta_f^\mathit{mod}$ be the modified parse tree where any leaf labeled with a constant simple function $e$ is replaced by the non-constant simple function $e'-\nu(c)$, where $e'=e{+}i$.

We will prove that $\felapse{f} = \mathit{qs}(\Delta^\mathit{mod}_f)$, which demonstrates that $\felapse{f}$ is quasi-simple as required. The proof is by induction on the rank of $f$. 
If $\rank{f}=0$, then there are two cases to consider.
\begin{itemize}
 \item If $\Delta_f$ is a leaf labelled with a constant simple function which for any $\nu \in X$ returns $e$ for some $e \in \Nat$, then for any $\nu \in X$:
\begin{align*}
\felapse{f}(\nu) = & \; t_\nu {+} f(\nu {+} t_\nu) \\
= & \;  t_\nu {+} e \tag{by definition of $\Delta_f$} \\
= & \; i {-} \nu(c) {+} e  & \tag{by definition of $t_\nu$} \\
= & \; e' {-} \nu(c) \tag{by definition of $e'$}
\end{align*}
which equals $\mathit{qs}(\Delta^\mathit{mod}_g)(\nu)$ as required.
\item If $\Delta_f$ is a leaf labelled with a simple function which for any $\nu \in X$ returns $e {-} \nu(c')$ for some $e \in \Nat$ and clock $c'$, then we have for any $\nu \in X$:
\begin{align*}
\felapse{f}(\nu) = & \;  t_\nu {+} f(\nu {+} t_\nu) \\
= & \; t_\nu {+} e {-} (\nu(c') {+} t_\nu) \tag{by definition of $\Delta_f$} \\
= & \;e {-} \nu(c') \tag{rearranging}
\end{align*}
which again equals $\mathit{qs}(\Delta^\mathit{mod}_g)(\nu)$ as required.
\end{itemize}
For the inductive step, suppose $\rank{f} = k{+}1$ for some $k \in \Nat$ and for any quasi-simple function of rank less than or equal to $k$ the result holds.
Since $\rank{f}=k{+}1$ there must be an operation $\mathit{op}$ (either min, max or convex combination) and integer $n$ such that $f$ is obtained by taking the $\mathit{op}$ of some quasi-simple functions $f_1,\ldots, f_{n}$, each of which has rank at most $k$. Now, for any $\nu \in X$:
\begin{align*}
 \felapse{f}(\nu) = & \; t_\nu {+} f(\nu {+}t_\nu)\\
= & \;  t_\nu {+} \mathit{op}(f_1(\nu {+} t_\nu), \ldots, f_n(\nu {+} t_\nu))\tag{by definition of $f$}\\
= & \; \mathit{op}(t_\nu {+} f_1(\nu {+} t_\nu), \ldots, t_\nu {+} f_n(\nu {+} t_\nu)) \tag{rearranging}\\
= & \; \mathit{op}(\felapse{f_1}(\nu),\ldots, \felapse{f_n}(\nu)) \tag{by definition of $\felapse{f_i}$ for $1 {\le} i {\le} n$}\\
= & \; \mathit{op}(\mathit{qs}(\Delta_{f_1}^\mathit{mod})(\nu),\ldots, \mathit{qs}(\Delta_{f_n}^\mathit{mod})(\nu))\tag{by the inductive hypothesis}\\
= & \; \mathit{qs}(\Delta^\mathit{mod}_{f})(\nu) \tag{by definition of $\Delta^\mathit{mod}_{f}$}
\end{align*}
This completes the induction step, and hence the lemma holds.\qed
\end{proof}
The next lemma states that resetting clocks preserves quasi-simplicity.
\begin{lemma}\label{lemma:quasi-reset}
For any region $\region$ and quasi-simple function $g : \region_C {\to} \Rplus^\infty$, the
function $\freset{g} : \region {\to} \Rplus^\infty$ defined by $\freset{g}(\nu) = g(\nu_C)$ is quasi-simple.
\end{lemma}
\begin{proof}
For a quasi-simple function $f$, let $\Delta^\mathit{mod}_f$ be the modified parse tree of $\Delta_f$ where a leaf labelled with a non-constant simple function which for any $\nu \in \region_C$ returns $e{-} \nu(c)$ for some integer $e$ and clock $c \in C$ is replaced with a leaf labelled by the constant function $e$. The proof follows by showing $\freset{f} = \mathit{qs}(\Delta^\mathit{mod}_f)$ for all quasi-simple functions $f$. This proof is by induction on the rank of $f$.

If $\rank{f}=0$, then $\Delta_f$ is a leaf and there are three cases to consider.
\begin{itemize}
 \item If $\Delta_f$ is a leaf labelled with a constant simple function which for any $\nu \in \region_C$ returns $e$ for some $e \in \Nat$, then for any $\nu \in \region$ by construction:
\begin{align*}
\freset{f}(\nu) = & \; f(\nu_C) \\
= & \; e \tag{by definition of $\Delta_f$} \\
= & \; \mathit{qs}(\Delta^\mathit{mod}_f§) \tag{by construction}
\end{align*}
\item If $\Delta_f$ is a leaf labelled with a simple function which for any $\nu \in \region_C$ returns $e {-} \nu(c')$ for some $e \in \Nat$ and clock $c' \not\in C$, then for any $\nu \in \region$:
\begin{align*}
\freset{f}(\nu) = & \; f(\nu_C) \\
= & \; e{-}\nu_C(c) \tag{by definition of $\Delta_f$} \\
= & \; e{-}\nu(c) \tag{since $c \not\in C$} \\
= & \; \mathit{qs}(\Delta^\mathit{mod}_f) \tag{by construction}
\end{align*}
\item If $\Delta_f$ is a leaf labelled with a simple function which for any $\nu \in \region_C$ returns $e {-} \nu(c)$ for some $e \in \Nat$ and clock $c \in C$, then we have for any $\nu \in \region$:
\begin{align*}
\freset{f}(\nu) = & \; f(\nu_C) \\
= & \; e{-}\nu_C(c) & \tag{by definition of $\Delta_f$} \\
= & \; e \tag{since $c \in C$} \\
= & \; \mathit{qs}(\Delta^\mathit{mod}_f) \tag{by construction}
\end{align*}
\end{itemize}
For the inductive step, suppose $\rank{f} = k{+}1$ for some $k \in \Nat$ and for any quasi-simple function of rank less than or equal to $k$ the result holds.
Since $\rank{f}=k{+}1$ there must be an operation $\mathit{op}$ (either min, max and convex combination) and integer $n$ such that $f$ is obtained by taking the $\mathit{op}$ of some quasi-simple functions $f_1,\ldots, f_{n}$, each of which has rank at most $k$. Therefore for any $\nu \in \region$ by construction:
\begin{align*}
\freset{f}(\nu) = & \; f(\nu_C)\\
= & \; \mathit{op}(f_1(\nu_C), \ldots, f_n(\nu_C))\tag{definition of $f$}\\
= & \; \mathit{op}(\freset{f_1}(\nu),\ldots, \freset{f_n}(\nu)) \tag{by definition of $\freset{f_i}$ for $1 {\le} i {\le} n$}\\
= & \; \mathit{op}(\mathit{qs}(\Delta_{f_1}^\mathit{mod})(\nu),\ldots, \mathit{qs}(\Delta_{f_n}^\mathit{mod})(\nu))\tag{by the inductive hypothesis}\\
= & \; \mathit{qs}(\Delta^\mathit{mod}_{f})(\nu) \tag{by definition of $\Delta^\mathit{mod}_{g}$}
\end{align*}
which completes the proof.\qed
\end{proof}
The following technical lemma will allow us to establish that, assuming quasi-simplicity in successor states, the players' optimal behaviour is to pick delays so that boundaries of regions are reached.
\begin{lemma}\label{lem:supinf-opt}
Let $f : X {\to} \Rplus^\infty$ be a quasi-simple function. For any $x\in X$ and $t\in \Rplus$ such that $x{+}t \in X$:
\begin{itemize}
\item $\sup_{t'\le t \wedge x+t' \in X} \left\{ t' {+} f(x{+}t') \right\} = t {+} f(x{+}t)$;
\item $\inf_{t'\ge t \wedge x+t' \in X} \left\{ t' {+} f(x{+}t') \right\} = t {+} f(x{+}t)$.
\end{itemize}
\end{lemma}
\begin{proof}
Consider any quasi-simple function $f : X {\to} \Rplus^\infty$ and clock $x$. It suffices to show that the function $t\mapsto t {+} f(x{+}t)$ is increasing.
Now for any $t_1,t_2 \in \Rplus$ such that $t_1\le t_2$ and $x{+}t_1,x{+}t_2 \in X$, we have:
\begin{align*}
 t_2 {+} f(x{+}t_2) = & \; t_1 {+} f(x{+}t_1) {+} ((t_2{-}t_1) {+} (f(x{+}t_2) {-} f(x{+}t_1)))\\
  \ge & \; t_1 {+} f(x{+}t_1)
\end{align*}
where the inequality follows since the term $(t_2{-}t_1) {+} (f(x{+}t_2) {-} f(x{+}t_1))$ is non-negative by the non-expansiveness of $f$ (see Lemma~\ref{lem:non-expansive}).\qed
\end{proof}

\subsection{Establishing correspondence of PTGA and boundary region abstraction}
\label{sect:proof-core}

\noindent
Having introduced quasi-simple functions and their properties, we will now show how they relate to PTGAs and
how they can be utilised to finish the proof of Theorem~\ref{theorem:main-thm-red}. The proof is notationally
heavy, and to alleviate some of the technical notation we first
introduce a number of functions (and properties of these functions)
that will allow us to abbreviate some of the notation. Intuitively, these functions
are counterparts to $\UVAL$ functions that in addition to an initial state also take the first action to be taken.
\begin{definition}\label{actions-def}
Let $n \in \Nat$, $\ell \in L$ and $\region \subseteq \inv(\ell)$. For $\star \in \{ \mMin,\mMax \}$ and $(a,\region',\oper) \in \RegB{A}_\star(\ell,\region)$, let $\UVAL^{n+1}_{\RegB{\pta}}((\ell,{\cdot},\region),(a,\region',\oper)) : \CLOS{\region} {\to} \Rplus$ be the function where for any $\nu \in \CLOS{\region}$ we have $\UVAL^{n+1}_{\RegB{\pta}}((\ell,\nu,\region),(a,\region',\oper))$ equal
\[ \begin{array}{c}
 \oper_{\nu + t\in \region'}  \left\{ t 
   + \sum\limits_{(\tilde\ell,\tilde\nu,\tilde\region)\in \hS} \hp_\star(\ell,\nu,\region), (a,\region',\oper))(\tilde\ell,\tilde\nu,\tilde\region) \cdot \UVAL^{n}_{\RegB{\pta}}(\tilde\ell, \tilde\nu, \tilde\region) \right\} \, .
\end{array}
\]
Furthermore, for $\nu \in \inv(\ell)$ and $(t,a) \in A_\mMIN \cup A_\mMAX$ such that $\nu{+}t\in \bar\region$ let:
\begin{align*}
  \UVAL_{\sem{\pta}}^{n+1}((\ell,\nu),(t,a)) = & \; t + \mbox{$\sum\limits_{(C, \ell') \in \powC \times L}$} \delta[\ell, a](C, \ell') \cdot \UVAL_{\sem{\pta}}^{n}(\ell',(\nu{+}t)_C)\\
  \UVAL_{\mix}^{n+1}((\ell, \nu), (t, a), \region) = & \; t + \mbox{$\sum\limits_{(C, \ell') \in \powC \times L}$} \delta[\ell, a](C, \ell') \cdot \UVAL_{\RegB{\pta}}^{n}(\ell',(\nu{+}t)_C,\region_C) \, .
\end{align*}
\end{definition}
Intuitively, $\UVAL^{n+1}_{\RegB{\pta}}((\ell,\nu,\region),(a,\region',\oper))$ corresponds to the optimal value in $(\ell,\nu,\region)$ when the
length of the horizon is $n{+}1$ and the first action performed is fixed to be $(a,\region',\oper)$. Similarly, $\UVAL_{\sem{\pta}}^{n+1}((\ell,\nu),(t,a))$ corresponds to the optimal value in $(\ell,\nu)$ for the length of the horizon $n{+}1$ when the first action performed is $(t,a)$. The definition of $\UVAL_{\mix}^{n+1}((\ell, \nu), (t, a), \region)$ gives an auxiliary function
which combines the values of $\sem{\pta}$ and $\RegB{\pta}$. Intuitively, it corresponds to taking a fixed action in $\sem{\pta}$, and then
transferring to $\RegB{\pta}$ for $n$ more steps.

Next we show that, within a region, the values in the BRA $\RegB{\pta}$ are quasi-simple when we restrict to a finite horizon reachability objectives.
To simplify the notation, we assume that in any state
player $\mMIN$ can pick at least one action, and that, for each action $a$ player $\mMIN$ can select, there exists an action $b$ player $\mMAX$ can select that is preferred, i.e. $b{=}\winner(a,b)$, and also an action $b$ that is not preferred, i.e. $a{=}\winner(a,b)$ (in addition, there can be actions $b$ not satisfying any of the two conditions). We refer to this assumption as {\em choice freedom}.
\begin{lemma}\label{actions-lem}
Assume $\pta$ is choice-free. For any $n \in \Nat$ and $s \in S$:
\[
\UVAL_{\sem{\pta}}^{n+1}(s) \; = \inf_{(t,a)\in A_\mMIN(s)} \max \left\{ \UVAL_{\sem{\pta}}^{n}(s,(t,a)) , \sup_{(t',b)\in A_\mMAX(s) \wedge t' \le t} \UVAL_{\sem{\pta}}^{n}(s,(t',b))\right\}
\]
Furthermore, for any $n \in \Nat$ and $\RegB{s} \in \RegB{S}$ we have that $\UVAL_{\RegB{\pta}}^{n+1}(\RegB{s})$ equals\footnote{Recall that $a {\le} b$ denotes the fact that $\hwinner(a,b) {=} a$.}:
\begin{align*}
& \min_{(a,\region,\inf) \in \hA_\mMIN(\RegB{s})} \Bigg\{ \UVAL^{n+1}_{\RegB{\pta}}(\RegB{s},(a,\region,\inf))  ,   \\ & \qquad 
\max_{\substack{(b,\region',\sup) \in \hA_\mMAX(\RegB{s})\\(b,\region',\sup)\le(a,\region,\inf)}}  \UVAL^{n+1}_{\RegB{\pta}}(\RegB{s},(b,\region',\sup)), \max_{(b,\region,\inf) \in \hA_\mMAX(\RegB{s})}\! \UVAL^{n+1}_{\RegB{\pta}}(\RegB{s},(b,\region,\inf)) \Bigg\}
\end{align*}
\end{lemma}
\begin{proof}
For $\UVAL_{\sem{\pta}}^{n+1}(s)$, the proof follows easily using Definition~\ref{opt-def}, choice-freedom, and properties of $\winner$.
For $\UVAL_{\RegB{\pta}}^{n+1}(\RegB{s})$ we use the definition of $\hwinner$ together with the fact that
$\UVAL^{n+1}_{\RegB{\pta}}(\RegB{s},(a,\region,\inf)) \le \UVAL^{n+1}_{\RegB{\pta}}(\RegB{s},(a,\region,\sup))$ for all $a\in \act$. The latter follows from Definition~\ref{brg-def}. \qed
\end{proof}
From now on, we will assume $\pta$ is choice-free. Note that this is purely a notational advantage, which will allow us to use Lemma~\ref{actions-lem}. The proofs we give can be easily extended to non-choice-free $\pta$ by omitting an appropriate part of the equations. For example, if $\act_\mMIN(s) {=} \emptyset$, then the first equation in the Lemma~\ref{actions-lem}
reduces to 
\[
\UVAL_{\sem{\pta}}^{n+1}(s) = \sup_{(t',b)\in A_\mMAX(s)} \UVAL_{\sem{\pta}}^{n}(s,(t',b)) \, .
\]
We now proceed with the following lemma which states that the $n$-step value functions on a BRA are regionally quasi-simple.
\begin{lemma}\label{lemma:pn-brg-qs}
For any $n \in \Nat$, $\ell \in L$ and $\region \in \Rr$ such that $\region \subseteq \inv(\ell)$, the function $\UVAL^{n}_{\RegB{\pta}}(\ell, {\cdot},\region)  : \CLOS{\region} {\to} \Rplus$ is quasi-simple.
\end{lemma}
\begin{proof}
Consider any $\ell \in L$ and $\region \in \Rr$ such that $\region \subseteq \inv(\ell)$.
We proceed by induction on $n \in \Nat$. For $n{=}0$, by Definition~\ref{bell-def} we have that $\UVAL^{0}_{\RegB{\pta}}(\ell, {\cdot},\region)$ is constant and equals 0, and hence quasi-simple.

Now suppose the claim holds for $n \in \Nat$. If $\ell \in L_F$, then $\UVAL^{n+1}_{\RegB{\pta}}(\ell, {\cdot},\region)$ is constant and equals 0, and hence quasi-simple. It therefore remains to consider the case when $\ell \not\in L_F$. By Definition~\ref{brg-def} for $\star \in \{ \mMin,\mMax \}$ we have $\RegB{A}_\star(\ell,\nu,\region) = \RegB{A}_\star(\ell,\nu',\region)$ for all $v,v' \in \CLOS{\region}$, hence we use $\RegB{A}_\star(\ell,\region)$ to denote $\RegB{A}_\star(\ell,\nu,\region)$ for any $v \in \CLOS{\region}$ and $\star \in \{ \mMin,\mMax \}$.

Using induction, Lemma~\ref{lemma:quasi-reset}, Lemma~\ref{lemma:quasi-next-region} and the quasi-simplicity of a convex combination of quasi-simple functions, it follows that the function $\UVAL^{n+1}_{\RegB{\pta}}((\ell,{\cdot},\region),\alpha) : \CLOS{\region} {\to} \Rplus$ given in Definition~\ref{actions-def} is quasi-simple for any $\star \in \{ \mMin,\mMax \}$ and $\alpha \in \RegB{A}_\star(\ell,\region)$.

Now, by Definition~\ref{bell-def}, for any $\nu \in \CLOS{\region}$:
\begin{align*}
\UVAL^{n+1}_{\RegB{\pta}}&(\ell,\nu,\region) \; = \; \min_{\alpha \in \hA_\mMIN(\ell,\region)} \max_{\beta \in \hA_\mMAX(\ell,\region)} \Bigg\{ \htau((\ell,\nu,\region), \alpha, \beta)  \\
& \hspace{3.1cm} + \sum_{(\tilde\ell,\tilde\nu,\tilde\region) \in \RegB{S}} \hp((\ell,\nu,\region), \alpha, \beta)(\tilde\ell,\tilde\nu,\tilde\region) \cdot \UVAL^{n}_{\RegB{\pta}}(\tilde\ell, \tilde\nu, \tilde\region) \Bigg\} \\
= &  \; \min_{\alpha \in \hA_\mMIN(\ell,\region)} \max \Bigg\{ \UVAL^{n+1}_{\RegB{\pta}}((\ell,\nu,\region),\alpha)  , \max_{\beta \in \hA_\mMAX(\ell,\region) \wedge \beta \le \alpha}\UVAL^{n+1}_{\RegB{\pta}}((\ell,\nu,\region),\beta) \Bigg\}
\end{align*}
by Definition~\ref{brg-def} and Definition~\ref{actions-def}.
Hence, $\UVAL^{n+1}_{\RegB{\pta}}(\ell,\cdot,\region) : \CLOS{\region} {\to} \Rplus$ equals an expression which takes the maxima and minima of quasi-simple functions, and therefore by Definition~\ref{quasi-def} is also quasi-simple. \qed
\end{proof}
The following lemma demonstrates that, for finite-horizon reachability-time objective, the values in the BRA and PTGA coincide.
\begin{lemma}\label{lemma:pn-qs}
For any $n \in \Nat$ and $s\in S$ we have $\UVAL^{n}_{\sem{\pta}}(s) = \UVAL^{n}_{\RegB{\pta}}(\RegB{s})$, and
hence $\UVAL^{n}_{\sem{\pta}}(\ell,\cdot) : \CLOS{\region} {\to} \Rplus$ is regionally quasi-simple for any $\ell \in L$ and $\region \in \Rr$ such that $\region \subseteq \inv(\ell)$.
\end{lemma}
\begin{proof}
Consider any $s{=}(\ell,\nu) \in S$. We proceed by induction on $n \in \Nat$. If $n{=}0$, then by Definition~\ref{bell-def} both
$\UVAL^{0}_{\sem{\pta}}(s)$ and $\UVAL^{0}_{\RegB{\pta}}(\RegB{s})$ equal 0, and hence the result holds.

Now assume that the lemma holds for some $n \in \Nat$. If $\ell \in L_F$ then 
\[
\UVAL^{n+1}_{\sem{\pta}}(s) = \UVAL^{n+1}_{\RegB{\pta}}(\RegB{s}) =0
\] 
and the result follows. It therefore remains to consider the case when $\ell \not\in L_F$. Using Lemma~\ref{actions-lem} we have:
\begin{equation}\label{eq:valnplus}
\!\!\UVAL_{\sem{\pta}}^{n+1}(s) = \!\!\inf_{(t,a)\in A_\mMIN(s)}\!\! \max \left\{  \UVAL_{\sem{\pta}}^{n+1}(s{,}(t{,}a)) , \!\!\sup_{\substack{(t',b)\in A_\mMAX(s) \\ t' \le t}}\!\! \UVAL_{\sem{\pta}}^{n+1}(s{,}(t'{,}b)) \right\} \, .
\end{equation}
Let $A_\mMAX^{\leq t}(s) = \{ (t',b) \in A_\mMAX(s) \mid t' {\le} t \}$ be the set of actions available to Player $\mMAX$ in $s$ with delay up to $t$, let $\Rr(\nu,t) = [[v],[v{+}t]]$ be the regions obtainable from $\nu$ by delaying at most $t$ time units, and $\act_\mMAX(\ell,\region) = \{ b \in \act_\mMAX \mid \region \subseteq E(\ell,b) \}$ be the set of actions of $\pta$ available when in location $\ell$ and region $\region$. It follows by Definition~\ref{ptgsem-def} that:
\begin{equation}\label{maxactions-eqn}
A_\mMAX^{\leq t}(s) = \bigcup_{\region \in \Rr(\nu,t)} \{ (t',b) \in \Rplus {\times} \act(\ell,\region) \mid  \nu {+} t' \in \region \wedge t' {\leq} t \}
\end{equation}
Furthermore, letting $\Rr(\nu) = \{ \region \in \Rr \mid \region \subseteq \inv(\ell) \wedge [\nu] \to^* \region \}$ be the set of
regions obtainable from $\nu$ by some delay and $\act_\mMIN(\ell,\region) = \{ a\in \act_\mMin \mid \region \subseteq E(\ell,a) \}$
the set of actions of player $\mMin$ available in location $\ell$ and region $\region$, again by Definition~\ref{ptgsem-def} we have:
\begin{equation}\label{minactions-eqn}
A_\mMIN(s) = \bigcup_{\region \in \Rr(\nu)} \{ (t,a) \in \Rplus {\times} \act_\mMIN(\ell,\region) \mid \nu {+} t \in \region \}
\end{equation}
Now, by Definition~\ref{actions-def}, letting $t^+_{\nu,\region'} = \sup \{t' \mid \nu{+}t' \in \region' \}$ we have:
\begin{align*}
 \lefteqn{\sup_{(t',b)\in A_\mMAX^{\leq t}} \UVAL_{\sem{\pta}}^{n+1}(s,(t',b))} \\
= & \; \sup_{(t',b)\in A_\mMAX^{\leq t}}\Bigg\{ t' + \sum_{(C, \ell') \in \powC \times L} \delta[\ell, b](C, \ell') \cdot \UVAL_{\sem{\pta}}^{n}(\ell',(\nu{+}t')_C) \Bigg\}
\\
= & \; \sup_{(t',b)\in A_\mMAX^{\leq t}} \Bigg\{ t' + \!\!\sum_{(C, \ell') \in \powC \times L}\!\! \delta[\ell, b](C, \ell') \cdot \UVAL_{\RegB{\pta}}^{n}(\ell',(\nu{+}t')_C,[(\nu{+}t')_C]) \Bigg\}
  \tag{by induction} \\
= & \; \max_{\substack{\region' \in \Rr(\nu,t) \\ b \in \act_\mMAX(\ell,\region')}}   \sup_{t' \le t \wedge \nu{+}t' \in \region'} \Bigg\{ t' {+}\!\!\!\! \sum_{(C, \ell') \in \powC \times L}\!\!\!\! \delta[\ell, b](C, \ell') {\cdot} \UVAL_{\RegB{\pta}}^{n}\!(\ell',(\nu{+}t')_C,\region'_C)\Bigg\} \tag{by (\ref{maxactions-eqn})} \\
= & \; \max_{\substack{\region' \in \Rr(\nu,t) \\ b \in \act_\mMAX(\ell,\region')}} \Bigg\{ \least{t}{\nu,\region'} +\\
 &\hspace{3cm}\sum_{(C, \ell') \in \powC \times L} \delta[\ell, b](C, \ell') \cdot \UVAL_{\RegB{\pta}}^{n}(\ell',(\nu{+}\least{t}{\nu,\region'})_C,\region'_C) \Bigg\} 
 \tag{since $\UVAL_{\RegB{\pta}}^{n}$ is quasi-simple and by Lemma~\ref{lem:supinf-opt}} \\
= & \; \max_{\substack{\region' \in \Rr(\nu,t) \\ b \in \act_\mMAX(\ell,\region')}} \UVAL_{\mix}^{n+1}(s, (\least{t}{\nu,\region'},b), \region') \tag{by Definition~\ref{actions-def}}
\end{align*}
Now substituting this into \eqref{eq:valnplus}, it follows that $\UVAL_{\sem{\pta}}^{n+1}(s)$ equals
\begin{align*}
& = \!\!\!\inf_{(t,a)\in A_\mMIN(s)} \max \left\{\UVAL_{\sem{\pta}}^{n+1}(s, (t, a)),\!\! \max_{\substack{\region' \in \Rr(\nu,t) \\ b \in \act_\mMAX(\ell,\region')}}\!\! \UVAL_{\mix}^{n+1}(s, (\least{t}{\nu,\region'}, b), \region') \right\} \\
&= \!\!\!\!\min_{\substack{\region \in \Rr(\nu) \\ a \in A_\mMIN(\ell,\region)}} \inf_{\nu + t \in \region} \max\left\{\UVAL_{\sem{\pta}}^{n+1}(s{,} (t{,} a)),\!\!\! \max_{\substack{\region' \in \Rr(\nu,t) \\ b \in \act_\mMAX(\ell,\region')}}\!\!\! \UVAL_{\mix}^{n+1}(s{,} (\least{t}{\nu{,}\region'}{,} b){,}\region')\right\} \tag{by (\ref{minactions-eqn})}
\end{align*}
For any $\region \in \Rr(\nu)$ and $(t,a) \in A_\mMIN(\ell,\region)$ the expression
\[
 \max_{\region' \in \Rr(\nu,t) \wedge b \in \act_\mMAX(\ell,\region')} \UVAL_{\mix}^{n+1}(s, (\least{t}{\nu,\region'}, b),\region')
\]
equals the maximum of
\[
 \max_{\substack{\region' \in \Rr(\nu,t) \setminus \region \\b \in \act_\mMAX(\ell,\region')}}\! \UVAL_{\mix}^{n+1}(s, (t^+_{\nu,\region'}, b),\region') \ \ \text{and}\ \max_{b \in \act_\mMAX(\ell,\region)}\! \UVAL_{\mix}^{n+1}(s, (\least{t}{\nu,\region}, b),\region)
\]
and both of these expressions decrease as $t$ decreases. Moreover, letting $t^-_{\nu,\region} = \inf \{t' \mid  \nu{+}t' \in \region \}$ and using Lemma~\ref{lem:supinf-opt}, Definition~\ref{actions-def}, and induction, we have:
\[
\inf_{\nu + t \in \region} \UVAL_{\sem{\pta}}^{n+1}(s,(t,a)) = \UVAL_{\mix}^{n+1}(s, (t^-_{\nu,\region}, a),\region)
\] 
Consequently, it follows that $\UVAL_{\sem{\pta}}^{n+1}(s)$ equals:
\begin{align*}
\min_{\substack{\region \in \Rr(\nu) \\ a \in A_\mMIN(\ell,\region)}} \max\Bigg\{\UVAL_{\mix}^{n+1}(s, (t^-_{\nu,\region}, a),\region), \max_{\substack{\region' \in \Rr(\nu,t) \setminus \region \\b \in \act_\mMAX(\ell,\region')}}  \UVAL_{\mix}^{n+1}(s, (t^+_{\nu,\region}, b),\region'), \\
 \max_{b \in \act_\mMAX(\ell,\region)}\! \UVAL_{\mix}^{n+1}(s, (t^-_{\nu,\region}, b),\region)\Bigg\}
\end{align*}
By definition of $t^-_{\nu,\region}$ and $t^+_{\nu,\region'}$ and Definition~\ref{actions-def} we have:
\begin{align*}
\UVAL_{\mix}^{n+1}(s,(t^-_{\nu,\region},a),\region)  = & \; \UVAL^{n+1}_{\RegB{\pta}}(\RegB{s}, (a, \region, \inf)) \\
\UVAL_{\mix}^{n+1}(s,(t^+_{\nu,\region},a),\region)  = & \; \UVAL^{n+1}_{\RegB{\pta}}(\RegB{s}, (a, \region, \sup))
\end{align*}
and hence, using definition of $\Rr(\nu,t)$, $\UVAL_{\sem{\pta}}^{n+1}(s)$ equals:
\begin{align*}
& \min_{(a,\region,\inf) \in \hA_\mMIN(\RegB{s})} \Bigg\{ \UVAL^{n+1}_{\RegB{\pta}}(\RegB{s},(a,\region,\inf))  ,   \\ & \qquad 
\max_{\substack{(b,\region',\sup) \in \hA_\mMAX(\RegB{s})\\(b,\region',\sup)\le(a,\region,\inf)}}  \UVAL^{n+1}_{\RegB{\pta}}(\RegB{s},(b,\region',\sup)), \max_{(b,\region,\inf) \in \hA_\mMAX(\RegB{s})}\! \UVAL^{n+1}_{\RegB{\pta}}(\RegB{s},(b,\region,\inf)) \Bigg\}
\end{align*}
which from Lemma~\ref{actions-lem} equals $\UVAL_{\RegB{\pta}}^{n+1}(\RegB{s})$, completing the proof.
\qed
\end{proof}
In the rest of this subsection we use the lemmas to prove properties of the infinite-horizon setting.
\begin{lemma}\label{lemma:lim-nonexp}
 The function $\lim_{n\rightarrow\infty} \UVAL_{\sem{\pta}}^{n}$ is regionally non-expansive.
\end{lemma}
\begin{proof}
The proof follows from Lemma~\ref{lemma:pn-qs} and the fact that a limit of non-expansive functions is a non-expansive function. \qed
\end{proof}
\begin{lemma}\label{lemma:lim-fixpoint}
The function $\lim_{n\rightarrow\infty}\UVAL_{\sem{\pta}}^{n}$ is a solution of the optimality equations $\UOpt_{\sem{\pta}}$.
\end{lemma}
\begin{proof}
Let $\Gamma = \lim_{n\rightarrow\infty}\UVAL_{\sem{\pta}}^{n}$ and for any $(\ell,\nu) \in S$ and $(t,a) \in A_\mMIN((\ell,\nu))$ let: 
 \[
  \Gamma((\ell, \nu), (t,a)) = t+ \sum_{(C, \ell') \in \powC \times L} \delta[\ell, a](C, \ell') \cdot\Gamma(\ell', (\nu{+}t)_C) \, .
 \]
By Definition~\ref{opt-def} and by using similar arguments as those from the proof of Lemma~\ref{actions-lem}, to prove the result it is sufficient to demonstrate that for any $s \in S$:
\begin{equation}\label{eq:gamma-fixpoint}
 \Gamma(s) = \inf_{(t,a)\in A_\mMIN(s) } \max \left\{\Gamma(s, (t,a)), \sup_{(t',b)\in A_\mMAX(s) \wedge t' \le t} \Gamma(s, (t',b))\right\} \, .
\end{equation}
Showing the left hand side is less than or equal to the right hand side follows easily from the monotonicity of the operator defining $\UOpt_{\sem{\pta}}$, i.e. of the operator
 $\mathcal{F} : (S {\to} \Rplus^\infty) {\to} (S{\to} \Rplus^\infty)$ given by 
 \[
  \mathcal{F}(\gamma)(s) = \inf_{(t,a)\in A_\mMIN(s) } \max \left\{\gamma(s, (t,a)), \sup_{(t',b)\in A_\mMAX(s) \wedge t' \le t} \gamma(s, (t',b))\right\}
 \]
and from the Knaster-Tarski fixpoint theorem which implies that for all ordinals $o_1 \le o_2$ we have $\mathcal{F}^{o_1}(0) \le \mathcal{F}^{o_2}(0)$
where $0$ is the lowest element in the complete lattice of functions $S {\to} \Rplus^\infty$ ordered with respect to $\le$.

We complete the proof of (\ref{eq:gamma-fixpoint}) by showing the left hand side is greater than or equal to the right hand side. Consider any $s \in S$. If $\Gamma(s)$ is infinite, then the result follows. On the other hand, if $\Gamma(s)$ is finite, it is sufficient to show that for any $\varepsilon {>} 0$:
 \[
  \Gamma(s) + \varepsilon \ge \inf_{(t,a)\in A_\mMIN(s)} \max \left\{\Gamma(s, (t,a)), \sup_{(t',b)\in A_\mMAX(s) \wedge t' \le t}\Gamma(s, (t',b))\right\} \, .
\]
We begin by selecting a finite sequence $t_1,\ldots, t_m$ of positive reals such that for any possible delay $t$ in $s{=}(\ell, \nu)$ there exists $t_i$ (denoted $\appr(t))$ with
 $[\nu{+}t] = [\nu{+}\appr(t)]$ and
 $|t {-} \appr(t)| \le \varepsilon/6$.
 Note that such a sequence $t_1,\ldots, t_m$ can always be selected as the clock values are bounded. By construction we have for any $t\in \Rplus$ and $C \subseteq \clocks$:
 \begin{equation}\label{eps-eqn}
 |(\nu{+}t) {-} (\nu{+}\appr(t))| \leq \varepsilon/6 \quad \mbox{and} \quad |(\nu{+}t)_C {-} (\nu{+}\appr(t))_C | \leq \varepsilon/6
 \end{equation}
 Now for any $(t,a)\in A_\mMIN\cup A_\mMAX$ we have:
\begin{align}
\lefteqn{\!\!\!\!\!\!\big|\Gamma(s, (t,a)) {-} \Gamma(s, (\appr(t),a))\big|} \nonumber \\
\le& \;  |t{-}\appr(t)| + \sum_{(C, \ell') \in \powC \times L} \delta[\ell, a](C, \ell') \cdot
     |\Gamma(\ell', (\nu{+}t)_C) {-} \Gamma(\ell', (\nu{+}\appr(t))_C)| \nonumber \\
\le& \;  |t{-}\appr(t)| + \sum_{(C, \ell') \in \powC \times L} \delta[\ell, a](C, \ell') \cdot
     | (\nu{+}t)_C {-} (\nu{+}\appr(t))_C| \tag{since $\Gamma$ is regionally non-expansive (Lemma~\ref{lemma:lim-nonexp})} \nonumber \\
\le& \;  |t{-}\appr(t)| + \sum_{(C, \ell') \in \powC \times L} \delta[\ell, a](C, \ell') \cdot
     \varepsilon/6  \tag{by (\ref{eps-eqn})} \nonumber \\
\le& \;  |t{-}\appr(t)| + \varepsilon/6 \tag{since $\delta[\ell, a](C, \ell')$ is a distribution} \nonumber \\
\le& \; \varepsilon/6 + \varepsilon/6 \tag{by construction of $\appr(t)$} \nonumber \\
= & \;  \varepsilon/3 \, . \label{eps1-eqn}
 \end{align}
By similar arguments (using Lemma~\ref{lem:non-expansive} and Lemma~\ref{lemma:pn-qs}) we can show that for any $n \in \Nat$:
  \begin{equation}\label{eps3-eqn}
  \left|\UVAL_{\sem{\pta}}^{n}(s,(t,a)) {-} \UVAL_{\sem{\pta}}^{n}(s,(\appr(t),a))\right| \leq \varepsilon/3
 \end{equation}
Since $\Gamma = \lim_{n\rightarrow\infty}\UVAL_{\sem{\pta}}^{n}$, for any $1 {\le} i {\le} m$ there exists $N_i \in \Nat$ such that:
\[  
\left|\Gamma(\ell', (\nu{+}t_i)_C) {-} \UVAL_{\sem{\pta}}^{n}(\ell', (\nu{+}t_i)_C)\right|\le \varepsilon/3
\]
for all $\ell' \in L$, $C \subseteq \clocks$ and $n {\geq} N_i$. Setting $N = 1 {+} \max_{1\le i \le m} N_i$, it follows that:
\begin{equation}\label{eps2-eqn}
\left|\Gamma(s,(\appr(t),a)) {-} \UVAL_{\sem{\pta}}^{n}(s,(\appr(t),a))\right|\le \varepsilon/3
\end{equation}
for all $(t,a)\in A_\mMIN\cup A_\mMAX$ and $n {\geq} N$. Now using (\ref{eps1-eqn}) we have for any $n {\geq} N$:
\begin{align}
 \Gamma(s, (t,a)) \le & \; \Gamma(s, (\appr(t),a)) + \varepsilon/3 \nonumber \\
  \le & \; \UVAL_{\sem{\pta}}^{n}(s,(\appr(t),a)) + \varepsilon/3 + \varepsilon/3 \tag{by (\ref{eps2-eqn})} \nonumber \\
  \le & \; \UVAL_{\sem{\pta}}^{n}(s,(t,a)) + \varepsilon/3 +  2\varepsilon/3 \tag{by (\ref{eps3-eqn})} \nonumber \\
  = & \; \UVAL_{\sem{\pta}}^{n}(s,(t,a)) + \varepsilon \label{eps4-eqn}
\end{align}
Finally, for any $n {\geq} N$ we have: 
\begin{align*}
\lefteqn{\inf_{(t,a) \in A_\mMIN(s)} \max \left\{\Gamma(s, (t,a)), \sup_{(t',b) \in A_\mMAX(s) \wedge t' \leq t} \Gamma(s, (t',b))\right\}}\\
\le & \; \inf_{(t,a) \in A_\mMIN(s)} \max \left\{\UVAL_{\sem{\pta}}^{n}(s,(t,a)) + \varepsilon, \sup_{\substack{(t',b) \in A_\mMAX(s)\\  t' \leq t}} \left\{\UVAL_{\sem{\pta}}^{n}(s,(t',b)) + \varepsilon \right\} \right\} \tag{by (\ref{eps4-eqn})}\\
= & \; \inf_{(t,a) \in A_\mMIN(s)} \max \left\{\UVAL_{\sem{\pta}}^{n}(s,(t,a)) , \sup_{(t',b) \in A_\mMAX(s) \wedge t' \leq t} \UVAL_{\sem{\pta}}^{n}(s,(t',b)) \right\} +  \varepsilon \tag{rearranging} \\
= & \;  \UVAL_{\sem{\pta}}^{n}(s) + \varepsilon \tag{by Lemma~\ref{actions-lem}} \\
\le & \; \Gamma (s) + \varepsilon &\tag{since $\UVAL_{\sem{\pta}}^{m}(s)\le\UVAL_{\sem{\pta}}^{m+1}$ for all $m$}
\end{align*}
which completes the proof. \qed
\end{proof}
We are now a few steps away from concluding the proof of the main result of the paper.
\begin{theorem}\label{thm:lim-is-val}
 $\UVAL_{\sem{\pta}}=\lim_{n\rightarrow \infty}\UVAL_{\sem{\pta}}^{n}$.
\end{theorem}
\begin{proof}
Using Lemma~\ref{lemma:ih-at-least-fh} it follows that $\UVAL_{\sem{\pta}}\ge\lim_{n\rightarrow \infty}\UVAL_{\sem{\pta}}^{n}$. On the other hand, Lemma~\ref{lemma:lim-fixpoint} states that $\lim_{n\rightarrow \infty}\UVAL_{\sem{\pta}}^{n}$ is a solution of the equations $\UOpt_{\sem{\pta}}$ and Lemma~\ref{lemma:value-under-fixpoint} states that $\UVAL_{\sem{\pta}} \le V$ for any solution $V$ of the equations $\UOpt_{\sem{\pta}}$.
Therefore, we have $\UVAL_{\sem{\pta}}\le\lim_{n\rightarrow \infty}\UVAL_{\sem{\pta}}^{n}$, which completes the proof. \qed
\end{proof}
The above theorem together with Lemma~\ref{lemma:pn-qs} tells us that, to compute $\UVAL_{\sem{\pta}}(s)$, it is sufficient to compute
$\lim_{n\rightarrow \infty}\UVAL_{\RegB{\pta}}^{n}(s)$, which is equal to $\UVAL_{\RegB{\pta}}(s)$ using results similar to~\cite{CH08,Con93}.
This completes the proof of Theorem~\ref{theorem:main-thm-red}.

\section{Conclusions}

\noindent
In this paper we introduced the reachability-time problem for PTGAs and showed that it is decidable and in NEXPTIME $\cap$ co-NEXPTIME. Our proof relies on an analysis of step-bounded value functions, showing that they are {\em quasi-simple} and non-expansive when infinite horizon is taken. This allows us to reduce the problem to the reachability-time problem on a finite abstraction. As opposed to the preliminary version of the work presented in \cite{FKNT10a}, the reduction works for an unrestricted class of PTGAs.

Although the computational complexity of solving games on timed automata is high, UPPAAL Tiga~\cite{BC+07} is able to solve practical
reachability and safety properties for timed games by using efficient symbolic
zone-based algorithms~\cite{CJLRR09,AH+09}.  A natural future direction is to investigate the possibility of devising similar algorithms for probabilistic timed games.

On the theoretical level, we plan to study if our approach can be utilised for extensions of reachability-time objectives by considering an appropriate class of reward-based properties.

\paragraph{Acknowledgements} This work is partly supported by ERC Advanced Grant
VERIWARE and EPSRC grant EP/M023656/1. Vojt\v{e}ch Forejt is also affiliated with FI MU, Brno, Czech Republic.

\bibliographystyle{elsarticle-num} 
\bibliography{papers}
\end{document}